%% file: PRL.tex
\documentclass[aps,prx,reprint,superscriptaddress,longbibliography,floatfix, nofootinbib]{revtex4-2}
\usepackage[T1]{fontenc}
\usepackage[utf8]{inputenc}
\usepackage[english]{babel}
\usepackage{graphicx}
\usepackage[table]{xcolor}
\usepackage{amsmath,amssymb,amsthm,mathtools,bm}
\usepackage{mathpazo}

\usepackage{tikz}
\usepackage{enumitem}
\usepackage{aliascnt,xparse}
\usepackage{microtype,needspace}
\PassOptionsToPackage{hyphens}{url}
\usepackage[hyperfootnotes=false,pdfpagelabels]{hyperref}
\hypersetup{
  colorlinks=true,linktocpage=true,breaklinks=true,
  urlcolor=blue!60!black,linkcolor=blue!60!black,citecolor=gray,
  bookmarksnumbered=true,bookmarksopen=true,bookmarksopenlevel=1,
  pdftitle={Deep thermalization with and without quantum chaos},
  pdfauthor={Soumik Ghosh, Arjun Mirani, Yihui Quek, Michelle Xu}
}
\usepackage[capitalise,compress]{cleveref}

\theoremstyle{plain}
\newtheorem{theorem}{Theorem}
\newaliascnt{lemma}{theorem}
\newtheorem{lemma}[lemma]{Lemma}
\aliascntresetthe{lemma}
\newaliascnt{corollary}{theorem}
\newtheorem{corollary}[corollary]{Corollary}
\aliascntresetthe{corollary}
\theoremstyle{definition}
\newaliascnt{definition}{theorem}
\newtheorem{definition}[definition]{Definition}
\aliascntresetthe{definition}

\crefname{theorem}{Theorem}{Theorems}
\crefname{lemma}{Lemma}{Lemmas}
\crefname{corollary}{Corollary}{Corollaries}
\crefname{definition}{Definition}{Definitions}
\crefname{section}{Section}{Sections}
\crefname{subsection}{Section}{Sections}
\newcommand{\lem}[1]{\hyperref[lem:#1]{Lemma~\ref*{lem:#1}}}

\newaliascnt{fact}{theorem}
\newtheorem{fact}[fact]{Fact}
\aliascntresetthe{fact}
\crefname{fact}{Fact}{Facts}

\NewDocumentCommand{\ket}{m}{\left|#1\right\rangle}
\NewDocumentCommand{\bra}{m}{\left\langle#1\right|}
\NewDocumentCommand{\ketbra}{m g}{%
  \IfNoValueTF{#2}{\ket{#1}\!\bra{#1}}{\ket{#1}\!\bra{#2}}}
\NewDocumentCommand{\braket}{m g}{%
  \IfNoValueTF{#2}{\left\langle#1\middle|#1\right\rangle}%
    {\left\langle#1\middle|#2\right\rangle}}
\newcommand{\norm}[1]{\left\lVert#1\right\rVert}
\newcommand{\abs}[1]{\left|#1\right|}
\DeclareMathOperator{\tr}{tr}
\DeclareMathOperator{\Tr}{Tr}
\DeclareMathOperator*{\E}{\mathbb{E}}
\renewcommand{\Pr}{\mathop{\mathbb{P}}}
\newcommand{\eps}{\varepsilon}
\renewcommand{\tilde}{\widetilde}

\let\paperincludegraphics\includegraphics
\RenewDocumentCommand{\includegraphics}{O{} m}{%
  \IfFileExists{#2}{\paperincludegraphics[#1]{#2}}{%
    \PackageWarningNoLine{PRL}{Missing figure #2; restore the original file}%
    \fbox{\parbox[c][1.1in][c]{0.85\linewidth}{\centering
      Missing figure file:\\[4pt]\texttt{\detokenize{#2}}}}}}

\allowdisplaybreaks[0]
\begin{document}
\title{Deep thermalization with and without quantum chaos} 

\date{\today}
\author{Soumik Ghosh}
\affiliation{Department of Computer Science$,$ University of Chicago}
\affiliation{Center for Theoretical Physics Massachusetts Institute of Technology}
\affiliation{Simons Institute for the Theory of Computing, UC Berkeley}
\author{Arjun Mirani}
\affiliation{Leinweber Institute for Theoretical Physics$,$ Stanford University}
\author{Yihui Quek}
\affiliation{School of Computer and Communication Sciences and School of Basic Sciences, École Polytechnique Fédérale de Lausanne}
\affiliation{Department of Mathematics, Massachusetts Institute of Technology}
\affiliation{Center for Theoretical Physics -- a Leinweber Institute, Massachusetts Institute of Technology}
\author{Michelle Xu}
\affiliation{Leinweber Institute for Theoretical Physics$,$ Stanford University}
\affiliation{Simons Institute for the Theory of Computing, UC Berkeley}
\affiliation{Joint Center for Quantum Information and Computer Science, University of Maryland}

\begin{abstract}
Can measuring part of a quantum system make the remainder more random than the whole? Deep thermalization is a phenomenon in which measurement on part of a thermalizing quantum many-body system produces a subsystem ensemble of states whose higher moments approximate those of Haar-random states, i.e. a quantum state design. In this work, we tightly limit how much randomness from the global state can be transferred to the subsystem in this way: without further assumptions, a global $k$-design degrades to at most a $k/2$-design on the unmeasured subsystem. We show this optimality bound by constructing a global $k$-design for which each sampled state hides a sign preference in its projected distribution that can always be discovered by analyzing moments of its projected ensemble greater than $k/2$.
Strikingly, we show this design loss can be reversed when the global design is generated by chaotic dynamics. Under the chaotic {\em constant-time} evolution of a Gaussian Unitary Ensemble (GUE) Hamiltonian, we find that the projected ensemble becomes exactly Haar-random in the thermodynamic limit --- even if the global state seems to be only a $\mathcal{O}(1)$-design. Moreover, we show this phenomenon does not require the full GUE structure: it persists under substantially weaker assumptions on the spectrum and eigenbasis, and admits finite-size versions with increasingly accurate projected designs as the measured subsystem grows. Our results identify global chaotic dynamics as a mechanism by which measurements can ``concentrate'' quantum randomness and place tight limits on when such amplification is possible.

\end{abstract}
\maketitle
{\renewcommand{\thefootnote}{}%
 \footnotetext{Authors are listed in alphabetical order.}}

Quantum randomness plays an increasingly central role across quantum information and many-body physics. In quantum information, quantum state and unitary designs provide a statistical notion of pseudorandomness: these ensembles aim to cheaply reproduce the moments of Haar randomness without implementing a genuinely Haar-random state or unitary. In many-body physics, higher moments of distributions over quantum states probe the degree to which chaotic evolution can scramble microscopic information, going well beyond what is visible from the reduced density matrix (i.e. the first moment) alone. 
This raises a basic question shared by the two communities: \emph{how much quantum randomness can be generated from simple dynamics and measurements?}


From the quantum-information perspective, there has been explosive recent interest in generating quantum state and unitary designs with minimal resources (see \cite{Brando2016randomcircuits,schuster2025random,LaRacuente_2026,cui2025unitary}). We briefly review the state-of-the-art. Without measurements, even allowing arbitrary long-range two-qubit gates and an unlimited number of ancillas, an $\epsilon$-approximate state $k$-design requires circuit depth at least $\Omega\!\left(\log k+\log\log(n/\eps)\right)$ to prepare \cite{cui2025unitary}.
Measurement-based quantum computation offers a different route: an entangled graph-state resource can be prepared in constant quantum depth, but this trades quantum depth for sequential classical feedforward, since later quantum operations must wait for earlier measurement outcomes to be processed. Remarkably, Ref.~\cite{TurnerMarkham} showed that adaptivity can be eliminated for design generation, at the cost of compiling the desired randomness into a specially constructed graph-state resource whose size and preparation cost scale as $\text{poly}(n,k,1/\eps)$. Random-matrix dynamics provide yet another route: Ref.~\cite{chen2024efficient} showed that concatenating only two independent Gaussian Unitary Ensemble (GUE) evolutions already produces an approximately Haar-random unitary, and used this phenomenon as a building block for efficient high-order unitary designs.

In this work, we show that an even simpler operation can amplify design quality: \textbf{a single constant-time chaotic Hamiltonian evolution followed by computational-basis measurement of an arbitrary subsystem can produce a much higher-order design on the complementary subsystem.}

How can such a simple protocol generate so much randomness? The answer lies in \emph{deep thermalization}: the phenomenon whereby measuring part of an evolving many-body system produces an ensemble on the remaining subsystem (see \Cref{fig:deepth}) whose higher-order moments approach those of the Haar ensemble
\cite{ho2022exact,Cotler2023EmergentQuantumStateDesigns,ippoliti2022solvable}. The resulting distribution over pure states on the unmeasured subsystem is known as the ``projected ensemble,'' and if the distribution forms a design, this is referred to as the projected design. In contrast, the standard definition of quantum thermalization only concerns itself with the reduced density matrix (i.e. first moment of the projected ensemble). Projected designs of this kind have now been found across a variety of many-body systems and quantum simulators
\cite{choi2023preparingrandomstates,Mark_2024}.

\begin{figure}[!tbp]
    \centering
    \includegraphics[width=0.9\linewidth]{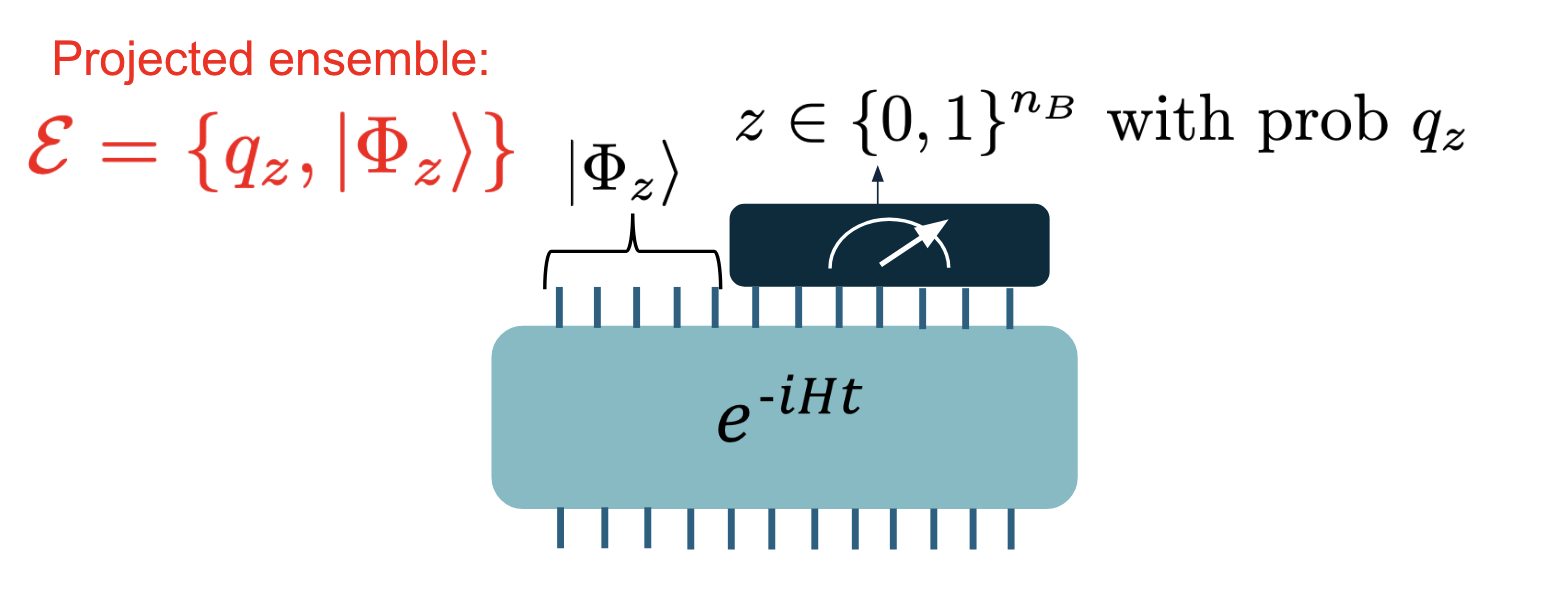}
    \caption{Projected ensemble setup: a global system is evolved under some unitary dynamics, in this figure Hamiltonian time evolution, and a subsystem of $n_B$ qubits is measured, producing the string $z\in \{0,1\}^{n_B}$ with probability $q_z$. The projected ensemble is the ensemble of states that remains on the unmeasured subsystem of size $n_A$ qubits.}
    \label{fig:deepth}
\end{figure}

A natural question is how the time evolution of the state before projection affects the ``amount of randomness'' in the projected ensemble. Implicit in the literature is the following intuition: \emph{chaotic global dynamics should produce exceptionally random projected ensembles}. For an intuitive example, Haar random evolution is both maximally chaotic under any metric and also deeply thermalizes. This intuition is reinforced by the Eigenstate Thermalization Hypothesis (ETH), which asserts that individual eigenstates of chaotic many-body systems already reproduce thermal expectation values of local observables \cite{BohigasGiannoniSchmit,DeutschETH,DAlessioETH}. Together, these results suggest that chaos progressively ``erases'' microscopic information and replaces it with universal statistics. 

It is thus natural to expect that projecting a chaotic system could produce exceptionally strong designs, which is also the central question of our work:
\begin{center}
\begin{minipage}{0.92\linewidth}
\centering
\emph{Does global quantum chaos boost the design quality of the projected ensemble?}
\end{minipage}
\end{center}
To make this question precise, we must first understand what projection can accomplish \emph{without} assuming anything on the dynamics generating the global state. We therefore begin by tightly determining the generic limits on how much design quality projection can preserve. Having established this baseline, we then return to chaotic dynamics and show that they can dramatically overcome these limits.

\paragraph{Upper bound: projection cannot generically preserve more than half the design order.} 

The first question we study is: if the global system forms a $k$ design and part of it is measured, does the design quality of the unmeasured subsystem increase or decrease? Here, we do not make any other assumptions on how the initial state was prepared. 

A first obstruction to design boosting was already known from stabilizer states. Random stabilizer states form an exact $3$-design globally, while measuring a subsystem in a stabilizer basis can produce a projected ensemble that is only a $1$-design, and remains far from a $2$-design \cite{VairogsYan2025,Mok2026NatureStingy}. More recently, Ref.~\cite{Mok2026NatureStingy} proved a general $2k\rightarrow k$ guarantee: a sufficiently accurate global $2k$-design yields an approximate projected $k$-design with high probability. The stabilizer example shows that this factor of two is optimal for $k=2$, but leaves open whether it is still optimal at higher moments. Indeed, an open question in Ref.~\cite{Mok2026NatureStingy} is whether the $2k\rightarrow k$ condition can be improved, for example to $k+O(1)\rightarrow k$.

Our first contribution resolves this question in the worst case. For every fixed $k$ and fixed unmeasured subsystem $A$, we construct an exact global $(2k-1)$-design whose projected ensemble fails to form even an approximate $k$-design: its $k$-th projected moment remains a constant trace-norm distance from Haar with probability approaching one as the measured subsystem grows. Thus, no global design order below $2k$ can generically guarantee an approximate projected $k$-design.

\paragraph{Lower bound: projection generically preserves half the design order.}

Our second result gives a matching lower bound on the projected design order. 
We show that a typical state drawn from a sufficiently accurate global $2k$-design produces, upon projection, an approximate $k$-design on the unmeasured subsystem with high probability. Equivalently, a global design of order $k$ generically retains at least $\lfloor k/2\rfloor$ moments after projection. Together with our first result, this pins down the worst-case loss in design order to precisely a factor of two.

A closely related $2k\rightarrow k$ achievability theorem subsequently appeared in Ref.~\cite{Mok2026NatureStingy}, approximately seven weeks after the first version of this work. Their result is formulated in the stronger relative-error notion of approximate designs, and achieves somewhat sharper quantitative guarantees on the required subsystem sizes. Our result instead works directly with additive trace-norm approximation and gives a particularly simple proof of the factor-of-two design-order guarantee.

\paragraph{Global chaos boosts design quality.} 
The preceding two results characterize what can be guaranteed from a global $k$-design: without any additional structure, a global $2k$-design guarantees only a projected $k$-design, and this factor-of-two loss is optimal in the worst case. We now ask whether additional knowledge about how the global state was generated can change this story.

We answer this question affirmatively for chaotic dynamics. We model the global dynamics by evolution under a Hamiltonian drawn from the Gaussian Unitary Ensemble (GUE), a canonical random matrix model of quantum chaos for systems without time-reversal symmetry \cite{BohigasGiannoniSchmit,Haake}. Here, we find a striking qualitative reversal of the generic behavior. At a discrete set of constant evolution times, the projected ensemble on the unmeasured subsystem becomes Haar-random in the thermodynamic limit of the measured subsystem: for every fixed $k$, its $k$-th moment converges exactly to the Haar $k$-th moment.

These special times are characterized by a simple spectral condition. They occur at the zeros of the first normalized trace moment,
\[
    \frac{J_1(2t)}{t}=0,
\]
so that controlling only the first spectral moment is sufficient for the projected ensemble to reproduce Haar moments of arbitrarily high order. Thus, chaotic dynamics can cause projection to dramatically concentrate the randomness on the unmeasured subsystem --- in effect {\em boosting} a seemingly global $\mathcal{O}(1)$ design into a local $\infty$-design.\footnote[1]{To the best of current literature understanding \cite{chen2024efficient,cotler2017chaos} and our discussion in SM~\ref{app:single_GUE}, GUE Hamiltonian evolution is not known to generate designs of greater than $\mathcal{O}(1)$ moments at $\mathcal{O}(1)$ times; however, we caveat that there does not currently exist a proof preventing it from doing so.}

\paragraph{Approximate global chaos still boosts design quality} While GUEs are mathematically appealing, their Haar-random eigenvectors have exponential circuit complexity, making them unphysical in real experiments. We prove that projected ensembles still form good designs, even when starting with more physically attainable assumptions on the dynamics of the original system. Specifically, we show that projected ensembles still form approximate $k$-designs as long as our Hamiltonian still has (i) a spectrum with some spectral zeroes and (ii) eigenbases that satisfy at least an approximate 4-design (as opposed to the Haar-random eigenbases of GUE evolution). We also provide finite $N$ corrections away from the thermodynamic limit.

~\\~\\
\noindent \textbf{Our setup} --- We start with a state $\Phi_{AB}\in\mathcal{H}_A\otimes\mathcal{H}_B$. In \autoref{fig:deepth}, this is the state post-unitary evolution by $e^{iHt}$. We then measure the $n_B$ qubits in the $\mathcal{H}_B$ system in the computational basis; our results do not depend specifically on the choice of basis, so we assume $z \in \{0,1\}^{n_B}$ without loss of generality. Corresponding to each measurement output $z$, which occurs with probability $q_z$, the system $\mathcal{H}_A$ collapses into a pure state $\Phi_z$. This defines an ensemble $\mathcal{E}=\{q_z,\ket{\Phi_z}\}$ called the \textit{projected ensemble.}  We also define the unnormalized state $\tilde{\Phi}_z$ for ease of analytical calculation. Mathematically, these quantities relate to each other as follows: 
\begin{align*}
\ket{\Phi_z} &= \ket{\tilde{\Phi}_z}/\sqrt{q_z}, \\
\ket{\tilde{\Phi}_z} &= (\mathbb{I}_A \otimes \bra{z}_B)\ket{\Phi}_{AB}, \\
q_z &= \braket{\tilde{\Phi}_z}{\tilde{\Phi}_z},
\qquad z \in \{0,1\}^{n_B}.
\end{align*}

In the establishing work on deep thermalization, Ref. \cite{Cotler2023EmergentQuantumStateDesigns} showed that $\mathcal{E}$ actually forms an approximate design. We call $\mathcal{E}$ an $\epsilon$-approximate projected ensemble $k$-design under the condition:
\begin{align*}
    \norm{\E_z (\ket{\Phi_z}\bra{\Phi_z})^{\otimes k}-\E_{\text{Haar}}(\ketbra{\phi})^{\otimes k}}_1\leq \epsilon.
\end{align*}
Note in particular that while the Haar average is taken over an entire state ensemble, the projected ensemble is \emph{only} averaged over $\Phi_z$ for some particular instance of $\Phi$, even if $\Phi$ comes from an ensemble of states or random evolution. These projected ensemble designs have been found in a variety of systems, including designs and 1-D Ising spin systems \cite{Cotler2023EmergentQuantumStateDesigns}, dual unitary systems of various kinds \cite{ho2022exact, ippoliti2022solvable, ippoliti2023dynamical, claeys2022emergent}, charge-conserving dynamics \cite{chang2025deep}, and more \cite{Mark_2024,lucas2023generalized, bhore2023deep, shrotriya2025nonlocality, liu2024deep}. These constructions have also been adapted for mixed states \cite{yu2025mixed, sherry2026deep}, pseudorandom states \cite{chakraborty2025fast}, and even finite-temperature systems via the conceptual shift to Scrooge designs \cite{ippoliti2022solvable,Mok2026NatureStingy}.

However, the quality of the projected ensemble designs found are vastly different from system to system. In \cite{Cotler2023EmergentQuantumStateDesigns}, the analysis showed that a global approximate state $k$-design could become an approximate projected $k/n_B$-design. This factor of $n_B$ is quite severe as the system size grows.
Meanwhile in select dual unitary setups, the projected ensemble can form a perfect design in the thermodynamic limit $n_B\rightarrow\infty$\cite{ho2022exact,ippoliti2022solvable}. 

Our following results, then, clarify the precise bounds of what kind of projected design can form from a global design, and motivate how the assumption of chaotic evolution may circumvent those restrictions.

~\\~\\
\noindent\textbf{$2k\rightarrow k$ projection is optimal} --- Our first result establishes an upper bound on the quality of the projected design that can be formed from a generic global design. Specifically, we construct a state $(2k-1)$-design that will have a projected ensemble that always fails to be an exact projected $k$-design. Formally, we prove:
\begin{theorem}\label{thm:optimality_2k-k_thm}
    For all fixed $k>0$, $N_A=2^{n_A}$, and $\varepsilon'>0$, there exists an $\varepsilon'$-approximate state $(2k-1)$-design $\Phi$ such that its projected ensemble does not form an $\varepsilon$-approximate $k$-design with probability:
    \begin{align*}\notag
    \operatorname{Pr}_{\Phi}&\left[\norm{\sum_{z\in\{0,1\}^{n_B}} q_z\,
    \bigl(\ket{\Phi_z}\!\bra{\Phi_z}\bigr)^{\otimes k}-\rho^{(k)}_{Haar}}_1\geq \varepsilon\right]\\
     &
     > 1-\mathcal{O}_{\varepsilon}\left(N_B^{-1}\right)
    \end{align*}
    for some constant $\varepsilon=\varepsilon(N_A,k)>0$.
\end{theorem}
We see in the thermodynamic limit $N_B\rightarrow \infty$, the probability is high and the concentration becomes absolute.

The details of this proof can be found in SM~\ref{sec:2k-k_secondproof}, but we provide a proof sketch here. To motivate our construction, we prove that a Haar random state $\Phi\in\mathcal{H}_{AB}$ can be represented as 
    \begin{align*}
        \ket\Phi = \sum_{z\in\{0,1\}^{n_B}} \sqrt{q_z}e^{i\theta_z}\ket{\psi_z}\ket z
    \end{align*}
with $\theta$ being a uniformly random phase, $q_z$ distributed as Dirichlet($N_A$,...,$N_A$), and $\psi_z$ as Haar random states on $\mathcal{H}_A$. We then construct a specific ensemble of states $\phi_z$ on $\mathcal{H}_A$ that possesses the same moment distribution as $\psi_z$ up until the $k$th moment. Specifically, we perturb the Haar state distribution with a hidden sign multiplied by a polynomial of degree $k$ that is orthogonal to all polynomials of lesser degree. Our design construction is then 
    \begin{align*}
        \ket\Phi = \sum_{z\in\{0,1\}^{n_B}} \sqrt{q_z}e^{i\theta_z}\ket{\phi_z}\ket z
    \end{align*}
    
In the analysis as a global design, the perturbation can only be discovered at the $2k$th moment because the correlation will not survive averaging over the hidden sign for a smaller moment. Moreover, the construction is an exact $(2k-1)$-design and satisfies any $\varepsilon'>0$, which is also why $\varepsilon$ does not depend on $\varepsilon'$. 

In the analysis as a projected design, however, there is no averaging over the hidden sign because we are only given an instance of $\Phi$. Then the only contribution to the moment operator is $\ketbra{\phi_z}{\phi_z}^{\otimes k}$. The difference from the Haar average emerges by construction, and the projected design concentrates around a non-Haar constant. We conclude then, that if we begin with a generic global design, then we can at best guarantee a projected design depreciated by half the moments.

~\\~\\
\noindent\textbf{$2k\rightarrow k$ projection is achievable} --- Our second result establishes a matching lower bound on the quality of projected designs coming from generic designs. Specifically, we show that when we start with a global state drawn from an approximate generic $2k$-design, with no assumptions about how that state is generated, we can get an approximate $k$-design using projection. Notably, we improve upon the previous analysis for generic designs in \cite{Cotler2023EmergentQuantumStateDesigns} by removing the dependence on system size. Formally, we prove: 
\begin{theorem}\label{thm:main_2k-k_thm}
 Let $N_A = 2^{n_A}$ and $N_B = 2^{n_B}$, such that the system sizes of A and B satisfy $n_B > (k+1)n_A$. Let $\Phi\in\mathcal{H}_{AB}$ be a state drawn from an $\varepsilon'$-approximate $2k$-design, for $$\varepsilon' < \mathcal{O}\left(\frac{1}{N_B^{2k}N_A}\right).$$ Then with probability $$1-\mathcal{O}(1/N_A^{(1/2)-\alpha}),$$ where $\alpha \in (0,\frac{1}{2})$ is a constant, the ensemble $\mathcal{E}$ is an $\varepsilon$-approximate quantum state $k$-design for $$\varepsilon<k/N_A^\alpha.$$
\end{theorem}
\noindent In particular, this theorem implies that with probability inverse-exponentially (in $n_A$) close to 1, the ensemble $\mathcal{E}$ is an inverse-exponentially $\varepsilon$-approximate state $k$-design, assuming $k = \mathcal{O}(\text{poly}(n)).$

At the level of design order, the factor-of-two loss is tight: our
previous construction shows that a global $(2k-1)$-design need not yield even
an approximate projected $k$-design, whereas Theorem~\ref{thm:main_2k-k_thm}
shows that a sufficiently accurate global $2k$-design does. 

The details of this proof can be found in SM~\ref{sec:2k-k_firstproof}, but we provide a proof sketch here.
Recall that the $k$-th moment of the projected ensemble can be written as
\begin{equation}\label{eq:Adef_main}
\begin{aligned}\notag
A(\ket{\Phi})
&= \sum_{z\in\{0,1\}^{n_B}} q_z\,
   \bigl(\ket{\Phi_z}\!\bra{\Phi_z}\bigr)^{\otimes k} \\
&= \sum_{z\in\{0,1\}^{n_B}}
   \frac{\bigl(\ket{\tilde{\Phi}_z}\!\bra{\tilde{\Phi}_z}\bigr)^{\otimes k}}
        {q_z^{\,k-1}} \, .
\end{aligned}
\end{equation}
\noindent where $q_z$ is the marginal probability of observing measurement $z$ when $B$ is measured, corresponding to normalized state $\Phi_z$ on subsystem $A$. 

We first bound the expected trace distance between the $k$-th moment of the projected ensemble and the $k$-th moment of the Haar measure on subsystem $A$, averaged over the global initial state drawn from an approximate $2k$ (state) design ensemble. This is done by introducing a polynomial function of $\ket{\Phi}$ that is simultaneously close in trace distance to the above rational function $A(\ket{\Phi})$ and to the $k$-th Haar moment, implicitly leveraging the strong concentration of the Haar measure. Specifically, we introduce the intermediate operator 
\begin{align*}
        B(\ket{\Phi}) =  \sum_{z\in\{0,1\}^{n_B}} \frac{\ket{\tilde{\Phi}_z}\!\bra{\tilde{\Phi}_z}^{\otimes k}}{\mathbb{E}_{\Phi\sim\text{Haar}}(q_z^{k-1})}.
    \end{align*} 
This operator has no physical significance in that it does not describe any quantum state that appears in the experiment and is merely mathematically convenient.

The distance between $A(\ket{\Phi})$ and $B(\ket{\Phi})$ can be controlled by the mean absolute deviation of $q_z^k$; conceptually, this is small thanks to the strong concentration of Haar.
Then the distance between $B(\ket{\Phi})$ and the $k$th Haar moment can be controlled by a calculation of the ratio of frame potentials.
The triangle inequality of both distances then controls the average trace distance between our projected ensemble and the Haar moment.

Lastly, we use Markov's inequality to turn this bound on the average trace distance into a probability bound on the trace distance for any \textit{individual} initial state drawn from the global $2k$-design ensemble. We can then conclude that with high probability, we find an approximate projected $k$-design.  

These last two results delineate all the limitations of starting with a generic design; we now turn our focus to circumventing this by adding the assumption of chaotic dynamics to the creation of our global designs.

~\\~\\
\noindent\textbf{Projecting a globally chaotic system gives a Haar-distributed state} --- We now present the first model of deep thermalization with a quantified (infinite) positive increase from the global design quality to the projected design quality.
We prove that in the thermodynamic limit, the GUE Hamiltonian projected ensemble becomes an exact Haar design at a series of discrete but infinite points, starting at $\mathcal{O}(1)$ times. Our best analyses show that the global unitary is only an $\mathcal{O}(1)$ (in system size) design at these early times.\footnotemark[\value{footnote}] This suggests that our projected ensemble experiences the effect of infinitely {\em boosting} the design quality.
Precisely, we show:
\begin{theorem} \label{thm:inf_thm1}
Let $G$ be an $N_{AB}$-dimensional random matrix drawn from the GUE that acts on $\mathcal{H}_{A}\otimes\mathcal{H}_B$. Then, the projected ensemble $\mathcal{E}$ obtained by evolving with $G$ as a Hamiltonian and then measuring out subsystem $B$ approaches an exact quantum state design in the thermodynamic limit for any time $t$ that satisfies $J_1(2t)/t = 0$. In particular, the first time at which this happens is of order $\mathcal{O}(1)$. 
\end{theorem}
The $J_1$ function is a Bessel function of the first kind, and its periodic roots dictates the discrete, but infinite series of times in which the design quality is ideal, as shown in \autoref{fig:bessel_plot}. The Bessel function given is equal to the first trace moment of our GUE evolution in the thermodynamic limit; as such, its roots correspond to when the first trace moment vanishes. Typically, for a global design, one would expect that {\em all} the trace moments need to vanish, but the magic of the projected ensemble lies in that we only seem to need to control the first trace moment.

The details of this proof can be found in SM~\ref{app:single_GUE}, but we provide the core analysis here. We proceed by examining the frame potential for the GUE Hamiltonian --- although the frame potential is a Schatten $2$-norm quantity and approximate state designs are usually defined in $1$-norm, finding an exact matching frame potential is necessary and sufficient for an exact design. 

Following methods in \cite{ippoliti2022solvable}, the frame potential allows us to leverage the replica trick: recall that by using standard techniques, the frame potential of a GUE Hamiltonian can be written as a sum of paths, given by
    \begin{align*}
        \mathbb{E}_{G}F^{(k)} = \mathbb{E}_{G}\sum_{z_1, z_2} q_{z_1} q_{z_2} \abs{\langle\Phi_{z_1}|\Phi_{z_2}\rangle}^{2k}.
    \end{align*}
After substituting in the marginal probabilities, this expression can be written as actions of replicas on the original state $\Phi.$ 

Inspired by \cite{chen2024efficient}, we treat the eigenvalues and eigenbasis distributions of GUE separately; the GUE is a pillar of random matrix theory, and spectral properties are well-studied and analytically tractable in the infinite-dimensional limit. The unitary-invariant eigenbasis allows the expression to be expanded using the Weingarten calculus. Lastly we find that only permutations contributing to a Haar design on the unmeasured system remain after the thermodynamic limit --- taking this limit is necessary for subleading contributions to fall away from the final answer.

\begin{figure}
    \centering
    \includegraphics[width=\linewidth]{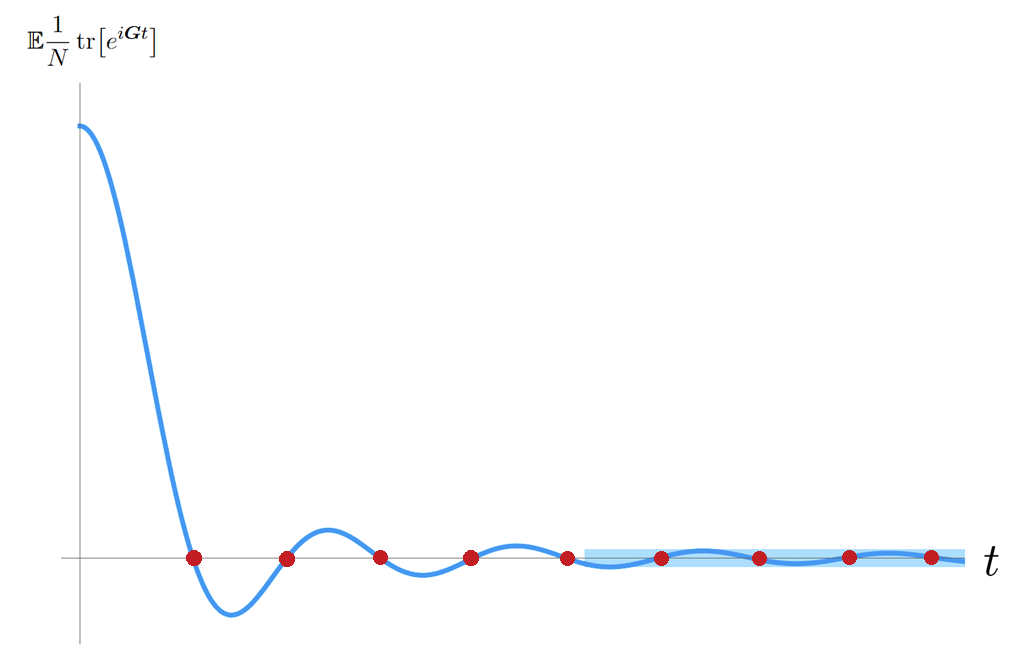}
    \caption{A plot of the first expected trace moment for the Gaussian Hamiltonian plotted against time $t$. At the zeros, marked in red, the projected ensemble of the Gaussian Hamiltonian is an exact design in the thermodynamic limit of the bath. This begins at early (i.e. constant) times and extends for infinite points. Past a certain time, all values of the expected trace moment are small, which is shown via the shaded region.}\label{fig:bessel_plot}
\end{figure}
~\\~\\
\noindent\textbf{Projecting a globally \emph{approximately} chaotic system still gives a good design} --- A skeptic might argue that GUE Hamiltonians are not physical. We therefore prove that our results also apply to systems with weaker assumptions than GUE dynamics. That is, the phenomenology demonstrated in \Cref{thm:inf_thm1} is robust: it requires neither the full GUE spectral statistics nor perfectly Haar-random eigenvectors. It also persists at finite system size (instead of the thermodynamic limit), where projection produces increasingly accurate high-order designs as the measured subsystem grows. That is, we are able to show:

\begin{theorem}[Robustness of design boosting, informal]\label{thm:inf_thm2}
The projected ensemble is also approximately Haar random under each of the following extensions of the idealized GUE setting for the global evolution ($V=UDU^\dagger$):
\begin{enumerate}
\item Beyond GUE eigenvalue statistics:
In the thermodynamic limit, the full GUE eigenvalue distribution is unnecessary: it suffices that the normalized first trace moment of (D) vanish and concentrate.

\item Beyond a Haar-random eigenbasis:
The Haar-randomness assumption on \(U\) may be replaced by the requirement that \(U\) be drawn from a sufficiently accurate approximate unitary \(4\)-design. The projected ensemble is then an \(\mathcal{O}(N_A^{-1})\)-approximate \(k\)-design with high probability \(1-\mathcal{O}(N_A^{-1})\).

\item Finite measured subsystem:
The phenomenon also persists away from the thermodynamic limit. For GUE evolution, if \(N_B>N_A^{k+3}\), then at every time \(t\) satisfying \(J_1(2t)/t=0\), the projected ensemble is an \(\mathcal{O}(N_A^{-1})\)-approximate \(k\)-design with probability \(1-\mathcal{O}(N_A^{-1})\).

\end{enumerate}
\end{theorem}

The analyses done to achieve these results use essentially the same techniques as the proof for \autoref{thm:inf_thm1}. We start again with expanding the frame potential of the projected ensemble. The proof of (1) goes through exactly the same as \autoref{thm:inf_thm1}. The proof for (2) uses H\"older's inequality to carefully bound the difference between the frame potentials; the requirement for 4-design bases comes from the fact that only two replicas of $U$ and $U^\dagger$ are used. The proof for (3) relies on a careful accounting of large $N_B$ behaviors in the Weingarten functions and GUE trace moments. 

As such, we precisely capture the error bounds on our approximate toy models. Such analysis show robustness, but one could wonder how physical such these approximate toy models are. Such questions are left to future work.

~\\~\\
\textbf{Discussion} --- We have proven an optimal bound for the design quality of a projected ensemble that comes from a generic global design; in particular, a global $2k$-design can at best form a projected $\lfloor k\rfloor$-design for $k$ any integer or half-integer. We have also shown that we can circumvent this by being more specific about the starting ensemble: from GUE evolution, we can get a projected design that seems infinitely boosted compared to the global design in the thermodynamic limit. 

One of the original conjectures about deep thermalization is that the formation of designs in the projected ensemble is part of the universality of quantum chaotic systems \cite{Cotler2023EmergentQuantumStateDesigns}. We have sharpened this by conjecturing that the \emph{boosting} of design quality in the projected ensemble is part of the universality of quantum chaotic systems. The combination of (i) decreasing quality in generic designs and (ii) robustness in the high design quality of GUE Hamiltonian and GUE-adjacent constructions seems to point to chaos being an important attribute in differentiating design behavior. Moreover, the proof techniques used in \autoref{thm:inf_thm1} hinge on the randomness of the eigenbases of the system evolutions, which suggests a bridge to the ETH.

Future directions include exploring this connection with ETH chaos to greater precision, and probing the projected design behavior in systems more local than GUE. Specifically, we would like to conduct analyses on how realistic it is for quantum many-body Hamiltonians to possess 4-design eigenbases. There are also outstanding connections to finite temperature Scrooge designs and applications to cryptography to explore.

~\\~\\
\emph{Acknowledgements:~}
We would like to thank Adam Bouland, Soonwon Choi, Laura Cui, Jordan Docter, Jeongwan Haah, Jonas Haferkamp, Patrick Hayden, Wai-Keong Mok, Bill Fefferman, Shantanav Chakraborty, Thomas Schuster, Wayne Weng, and Yuzhen Zhang for helpful conversations and discussions. M.X. acknowledges support from a NSF Graduate Research Fellowship. A.M. acknowledges support from the Herb and Jane Dwight Stanford Graduate Fellowship. Y.Q. is supported by a collaboration between the US DOE and other Agencies. This material is based upon work supported by the U.S. Department of Energy, Office of Science, National Quantum Information Science Research Centers, Quantum Systems Accelerator. 
This work was done in part while some of the authors were visiting the Simons Institute for the Theory of Computing and the Challenge Institute for Quantum Computation at UC Berkeley.

\bibliographystyle{apsrev4-2}
\bibliography{bibliography.bib}

\clearpage

\input{PRL_appendices.tex}

\end{document}

%% file: PRL_appendices.tex
\onecolumngrid

\section*{Supplemental Materials}
In the supplemental materials, we provide precise and in-depth proofs for our theorems in the text.
\appendix
\setcounter{theorem}{0}
\renewcommand{\theHtheorem}{appendix.\arabic{theorem}}
\section{The frame potential}\label{sec:frame}
We collect here some observations that we will return to repeatedly. For any probability distribution $\mu$ over $n$-qubit states, let
\begin{align}
    \rho_{\mu}^{(k)} &:= \mathbb{E}_{\Phi\leftarrow \mu} \left[\ketbra{\Phi}{\Phi}^{k}\right] \label{eq:moment_k}\\
    \rho_{\text{Haar}}^{(k)} &:= \frac{\Pi_{\mathrm{Sym}^k(\mathbb{C}^N)} }{\dim \mathrm{Sym}^k(\mathbb{C}^N)} \label{eq:moment_k_H}
\end{align}
In the setting of deep thermalization, we will often take $\rho_{\mu}^{(k)}$ to be the following ensemble:
\begin{equation}
\rho^{(k)} := \mathbb{E}_{\Phi\leftarrow \mu} \left[\sum_{z\in \{0,1\}^{n_B}} q_z \ketbra{\Phi_z}{\Phi_z}^{\otimes k}\right] 
\end{equation}

Since we will often be interested in bounding the trace norm distance between expected moment operators, we state the following useful inequalities:
\begin{lemma}[Bounding the trace distance from any $k$-th moment operator to the $k$-th Haar moment]
For any $k$ and $\mu$, let $\rho^{(k)} := \rho_{\mu}^{(k)}$ and $\rho^{(k)}_{Haar}$ be the moment operators defined in \Cref{eq:moment_k} and \Cref{eq:moment_k_H}. 
Define the quantity
\begin{align}
    \Delta^{(k)} := \frac{\norm{\rho^{(k)} - \rho^{(k)}_{Haar}}_2 }{\norm{\rho^{(k)}_{Haar}}_2}.
\end{align}
Then 
\begin{align}\label{eq:useful}
\boxed{\;
        \E \norm{\rho^{(k)}- \rho_{\text{Haar}}^{(k)}}_1
        \;\leq\;
        \E \left[\Delta^{(k)}\right]
        \;\leq\;
        \sqrt{\E\!\left[(\Delta^{(k)})^2\right]}
      \;}
\end{align}
This bound also applies when $\rho^{(k)}$ is replaced by any operator in the symmetric subspace, $\mathrm{Sym}^k(\mathbb{C}^N)$.
\end{lemma}
\begin{proof}
    The second inequality of \Cref{eq:useful} is simply Jensen's inequality.
    The first inequality of \Cref{eq:useful} follows from an application of Cauchy-Schwarz, i.e. we use the inequality $\|X\|_1 \leq \sqrt{\operatorname{rank}(X)}\|X\|_2$. First, we note that both $\rho^{(k)}$ and $\rho_{\text{Haar}}$ live in the symmetric subspace, therefore $X$ will also live on this subspace. Therefore, the multiplicative blowup when converting from $1\rightarrow 2$ norm is
    \begin{equation}
        \text{rank}(\Pi_{\mathrm{Sym}^k(\mathbb{C}^N)}) =  \dim \mathrm{Sym}^k(\mathbb{C}^N) = \binom{N+k-1}{k} = \frac{1}{\norm{\rho_{\text{Haar}}^{(k)}}_2^2}.
    \end{equation}
    We point out that a naive application of Cauchy-Schwarz on the full space would have resulted in a significantly worse multiplicative blowup of $\sqrt{N^k}$; therefore, this Lemma allows for a much tighter bound on the trace distance than vanilla Cauchy-Schwarz.  
\end{proof}

In fact, this $\Delta^{(k)}$ quantity may be re-stated as a function of the purity of the ensemble -- a quantity known as the frame potential. Via this re-writing, bounding trace norm distance from Haar moment operators reduces to computing purities. 

\begin{definition}[Frame Potential]
For any ensemble of states represented by $\rho$, the corresponding frame potential is defined as
\begin{align}
    F^{(k)} := \norm{\rho^{(k)}}^2_2 = \tr\left((\rho^{(k)})^2\right).
\end{align}
\end{definition}
Letting $s:= \dim \mathrm{Sym}^k(\mathbb{C}^N)$, 
\begin{align}
\left\|\rho^{(k)}-\rho_{\text {Haar }}^{(k)}\right\|_2^2=\operatorname{Tr}\left[\left(\rho^{(k)}\right)^2\right] -2 \operatorname{Tr}\left[\rho^{(k)} \rho_{\text {Haar }}^{(k)}\right]+\operatorname{Tr}\left[\left(\rho_{\text {Haar }}^{(k)}\right)^2\right] .
\end{align}
One may then calculate that $F^{(k)}_{Haar} = \operatorname{Tr}\left[\left(\rho_{\text {Haar }}^{(k)}\right)^2\right] = 1/s$ and $\operatorname{Tr}\left[\rho^{(k)} \rho_{\text {Haar }}^{(k)}\right] = 1/s$, so that ultimately
\begin{align}    
    \Delta^{(k)} = \frac{\norm{\rho^{(k)} - \rho^{(k)}_{Haar}}_2 }{\norm{\rho^{(k)}_{Haar}}_2} = \frac{\sqrt{F^{(k)}- F^{(k)}_{Haar}}}{\sqrt{F^{(k)}_{Haar}}} = \left(\frac{F^{(k)}}{F^{(k)}_{Haar}} - 1 \right)^{1/2}.
\end{align}

We observe that any state $k$-design $\mu$ has $\Delta^{(k)}=0$. It turns out that this is also a sufficient condition for $\mu$ to be an exact state $k$-design.
\begin{fact} If $\Delta^{(k)}=0$, then
\begin{align}
     \rho_{\mu}^{(k)} = \rho_{Haar}^{(k)}.
\end{align}
In other words, for any $t\leq k,$ if $\Tr_{k-t}(\rho_{\mu}^{(k)})  = \Tr_{k-t}(\rho_{Haar}^{(k)})$, then
\begin{align}
    \rho_{\mu}^{(t)} = \rho_{Haar}^{(t)}.
\end{align}
\end{fact}
This is exactly the condition that a $k$-design must satisfy.
Proving that an ensemble of states is an exact state $k$-design via checking that $\Delta^{(k)}=0$ is known as `the method of frame potentials'. We will use this method extensively.
\section{Facts about GUE}\label{app:gue_facts}
The following proofs require a number of lemmas about the properties of exponentiated GUEs. All of them along with proofs can be found in \cite{chen2024efficient} but we restate them below for pedagogy and ease of reference.

The GUE ensemble comprises of $N\times N$ Hermitian matrices with i.i.d. Gaussian entries. The spectral distribution of GUE follows the well-known semi-circle law, and its expected trace moments in the infinite $N$ limit are given by the Catalan numbers. Analysis of the expected trace moments of the exponentiated form gives a spherical bessel as its infinite limit, and then finite $N$ corrections result from controlling the spectral radius. As such, the following bound holds:
\begin{lemma}\label{lem:exp_gue_tr_avg}
    Let $G$ be an $N$-dimensional random matrix drawn from the GUE. Then there exists some constant $K_0>0$, such that for all $t>1$,
    \begin{align*}
        \left|\mathbb{E}\frac{1}{N}\tr[e^{i G t}]\right| \leq  \left|\frac{J_1(2t)}{t}\right|+ \frac{K_0t}{N}.
    \end{align*}
\end{lemma}
In the infinite limit, we find naturally that at the roots of the spherical Bessel $J_1$, the expected trace moments vanish.
\begin{lemma}\label{lem:good_t_exists}
    Let $G$ be an $N$-dimensional random matrix drawn from the GUE. Then there exists an infinite number of $t\in\mathbb{R}$ such that as $N\rightarrow\infty$,
    \begin{align}
        \mathbb{E}\frac{1}{N}\tr[e^{i G t}] \rightarrow 0.
    \end{align}
    There does not exist a constant greater than all the values of $t.$
\end{lemma}
There are other properties of the spherical bessel that will be useful to us. In particular, some bounds:
\begin{lemma}\label{lem:j1_bounds}
    For $x\geq 1$, $|J_1(x)|<\sqrt{\frac{1}{x}}$.
\end{lemma}
Lastly, we will need to know concentration inequalities for GUE. Due to the Gaussian nature of the entries, our exponentiated Gaussians concentrate quite well.
\begin{lemma}\label{lem:single_exp_gue_concentration}
    Let $G$ be an $N$-dimensional random matrix drawn from the GUE. Then,
\begin{align}
    \operatorname{Pr}[\left|\frac{1}{N}\tr[e^{i G  t}] - \E\frac{1}{N}\tr[e^{i G  t}]\right| \geq s ] &\leq \exp(-\frac{Ns^2}{2t^2}).  
\end{align}
\end{lemma}
More generally, Gaussian vectors have the following concentration property \cite{boucheron2003concentration}:
\begin{lemma}[Gaussian concentration inequality]\label{lem:general_gaussian_concentration}
Let $\bm g=\left(g_i\right)$ where $g_i \sim \mathcal{N}(0,1)$.
$$
\operatorname{Pr}[|f(\bm g)-\mathbb{E} f(\bm g)| \geq t] \leq e^{-t^2 / 2 L^2}, \text { $f$ is L-lipshitz }
$$
\end{lemma}

The following theorem from \cite{chen2024efficient} concerns using small trace moments to prove design quality.
\begin{theorem}[From small moments to Haar]\label{thm:moments_implies_Haar}
    Let $\bm{W}$ be a random unitary that is invariant under unitary conjugation:
    $$\bm{W} \stackrel{dist}{\sim} \bm{U} \bm{D} \bm{U}^{\dagger} \quad \text{for deterministic unitary $\bm{D}$ and Haar $\bm{U}$}.$$ Consider its first $T$ normalized trace moments $\alpha_p :=  \frac{1}{N} \tr(\bm{W}^p)$ for $1 \leq p \leq T$, and let $\vec{\alpha_T} := (\alpha_1, \dots,  \alpha_T)$.
    There is an absolute constant $C>0$ such that if $T$ is small enough, $\frac{T^5}{N} \leq \frac{C \delta}{\sqrt{\log{(4/\delta)}}}$, then 
    \begin{align*}
         \norm{\vec{\alpha_T}}_1 <  \frac{\delta}{32\cdot T^{7/2} }\quad \text{implies}\quad \abs{\mathbb{E} \mathcal{A}(\bm{W}) - \mathbb{E} \mathcal{A}(\bm{U}) } < \delta.
    \end{align*}
    That is, small trace moments implies that $\bm{W}$ is an adaptive, additive-error $\delta$-approximate $T$-query unitary design: no quantum algorithm making at most $T$ adaptive queries to $\bm W$ can distinguish it from a Haar-random unitary with advantage greater than $\delta$.
\end{theorem}

Finally, we will prove the following theorem about the design quality of a single exponentiated Gaussian. 

\begin{lemma}\label{lem:singleG_moments}
    Let $\bm{G}\in \mathrm{GL}(N)$ be a GUE random matrix. Then the unitary ensemble
    \begin{align}
        \bm{W} := e^{iG t}
    \end{align}
    is an adaptive $\epsilon$-approximate $k$-design where the error is a small constant $1>\epsilon=\mathcal{O}(1)$ for $N = \Omega(k^{47/3}) \approx \Omega(k^{15.67})$, $t=\Omega(k^{7/3})$.
\end{lemma}
\begin{proof}
    Let $\vec{\alpha}_k$ be the vector consisting of the first $k$ normalized trace moments of $\bm{W}$: $\vec{\alpha} = (\alpha_1,\dots,\alpha_k)$, where for each $p \in \{1,\dots,k\}$, $\alpha_p = \frac{1}{N}\tr\bm{W}^p$. The goal is to bound $\norm{\vec{\alpha}_k}_1$ with high probability, such that \autoref{thm:moments_implies_Haar} applies, and bound the failure probability of $\norm{\vec{\alpha}_k}_1$ exceeding that bound. 
    
    We can bound $\norm{\vec{\alpha}_k}_1$ with high probability using the Gaussian concentration property of \autoref{lem:general_gaussian_concentration}. To bound the Lipschitz constant $L(\norm{\vec{\alpha}_k}_1)$ of $\norm{\vec{\alpha}_k}_1$, note that $L(\alpha_p) \leq \frac{pt}{\sqrt{N}} \leq \frac{kt}{\sqrt{N}}$ for $1 \leq p \leq k$, so $L(\norm{\vec{\alpha}_k}_1) \leq \frac{k^2t}{\sqrt{N}}$.

    Then by the Gaussian concentration property,
    \begin{equation}
        \operatorname{Pr}[|\norm{\vec{\alpha}_k}_1 - \mathbb{E}\norm{\vec{\alpha}_k}_1 | \geq s ] \leq \exp(-\frac{Ns^2}{2k^4t^2}),
    \end{equation}
    which, by choosing the value of $s$ appropriately, can be re-written as 
    \begin{equation}\label{eq:sg_one_norm_concentration}
        \operatorname{Pr}[|\norm{\vec{\alpha}_k}_1 - \mathbb{E}\norm{\vec{\alpha}_k}_1 | \geq \frac{k^2t \sqrt{2\log(\frac{2}{\epsilon})}}{\sqrt{N}}] \leq \frac{\epsilon}{2}
    \end{equation}
    for $1 > \epsilon = \mathcal{O}(1)$.

    Now it follows that with probability (at least) $1-\frac{\epsilon}{2}$, which is the high-probability case, $\norm{\vec{\alpha}_k}_1$ is bounded as follows:
    \begin{equation}\label{eq:one_norm_bound}
        \norm{\vec{\alpha}_k}_1 \leq \mathbb{E}\norm{\vec{\alpha}_k}_1 + \frac{k^2t \sqrt{2\log(\frac{2}{\epsilon})}}{\sqrt{N}}.
    \end{equation}

    We can bound the right hand side by bounding $\mathbb{E}\norm{\vec{\alpha}_k}_1$. To do so, we will need the Gaussian concentration property for each $\alpha_p$:
    \begin{equation}
        \operatorname{Pr}[|\alpha_p - \mathbb{E}\alpha_p | \geq \gamma ] \leq \exp(-\frac{N\gamma^2}{2p^2t^2}),
    \end{equation}

    Using this, we can bound $\mathbb{E}|\alpha_p|$ for each $p$:
\begin{align}
    \mathbb{E}|\alpha_p| &= \mathbb{E}[\abs{\alpha_p} | |\alpha_p - \mathbb{E}\alpha_p| > \gamma]\cdot P[|\alpha_p - \mathbb{E}\alpha_p| > \gamma] + \mathbb{E}[\abs{\alpha_p} | |\alpha_p - \mathbb{E}\alpha_p| \leq \gamma]\cdot P[|\alpha_p - \mathbb{E}\alpha_p| \leq \gamma]\\
        &\leq \exp(-\frac{N\gamma^2}{2p^2t^2}) + (|\mathbb{E}\alpha_p| + \gamma)(1-\exp(-\frac{N\gamma^2}{2p^2t^2}))\\
        &\leq \exp(-\frac{N\gamma^2}{2p^2t^2}) + \left(\abs{\frac{J_1(2pt)}{pt}} + \frac{K_0pt}{N} + \gamma\right)\left(1-\exp(-\frac{N\gamma^2}{2p^2t^2})\right)
    \end{align}
    where the second inequality follows from \autoref{lem:exp_gue_tr_avg}, $J_1$ is the spherical Bessel function and $K_0$ is some $\mathcal{O}(1)$ constant. 

    Now we can bound $\mathbb{E}\norm{\vec{\alpha}_k}_1$:
    \begin{align}
        \mathbb{E}\norm{\vec{\alpha}_k}_1 &\leq \sum_{p=1}^k \left(\abs{\frac{J_1(2pt)}{pt}} + \frac{K_0pt}{N} + \gamma\right) + \text{subleading}\\
        &\leq \left(\frac{c}{\sqrt{2t^3}} + \frac{K_0k^2t}{N} + k\gamma \right) + \text{subleading}
    \end{align}
    where $c$ is an $\mathcal{O}(1)$ constant. Inserting this into the upper bound in \autoref{eq:one_norm_bound}, 
    \begin{equation}\begin{aligned}
        \norm{\vec{\alpha}_k}_1 \leq \frac{c}{\sqrt{2t^3}} + \frac{K_0k^2t}{N} + k\gamma  + \frac{k^2t \sqrt{2\log(\frac{2}{\epsilon})}}{\sqrt{N}} + \text{subleading}.
    \end{aligned}\end{equation}

    To obtain the desired condition $\norm{\vec{\alpha}_k}_1 \leq \frac{\epsilon/2}{32k^{7/2}}$, it suffices to have $t \geq \Omega(k^{7/3})$, $N > \Omega(k^{47/3})$ and $\gamma \leq \mathcal{O}(k^{-9/2})$. By \autoref{thm:moments_implies_Haar}, this implies that the probability of distinguishing $\bm{W}$ from a Haar random unitary by any $k$-query decision algorithm is bounded by $\frac{\epsilon}{2}$, that is, $\abs{\mathbb{E} \mathcal{A}(\bm{W}) - \mathbb{E} \mathcal{A}(\bm{U})} \leq \frac{\epsilon}{2}$.

     Finally, by \autoref{eq:sg_one_norm_concentration}, the low-probability case in which $\norm{\vec{\alpha}_k}_1 > \frac{\epsilon/2}{32k^{7/2}}$ occurs with probability at most $\frac{\epsilon}{2}$. So the total probability of distinguishing $\bm{W}$ from Haar by any $k$-query decision algorithm is bounded by $\epsilon$. Setting $1 > \epsilon = \mathcal{O}(1)$, we obtain the stated result. 
\end{proof}

This is an adaptive, additive error unitary design analysis for $e^{iGt}$; it is not a tight bound on error but is the best analysis we have currently for adaptive additive error. This then translates into states formed by such evolutions straightforwardly.


\section{Single Gaussian Hamiltonian projected ensemble}\label{app:single_GUE}

In this section we prove that the projected ensemble of Gaussian Hamiltonian evolution (in other words, our exponentiated GUE) is a $k$-design at a series of discrete times, some of which are quite early, in the thermodynamic limit of the measured bath size. This type of result in the thermodynamic limit is very akin to the result from \cite{ho2022exact} on dual unitary circuits and \cite{ippoliti2022solvable} on random quantum circuits.

It holds that $F^{(k)}\geq F^{(k)}_{Haar}$. Matching the Haar frame potential exactly is a sufficient and necessary condition for an exact $k$-design. This is the method we will take for our projected ensemble below.

\begin{theorem}\label{thm:projected_designs_GUE}
    Let $G$ be an $N_{AB}$-dimensional random matrix drawn from the GUE that acts on $\mathcal{H}_{A}\otimes\mathcal{H}_B$. Then consider the projected ensemble formed from time evolution under this Gaussian Hamiltonian on some initial state $\ket\phi$ and measurement in the computational basis on $n_B$ qubits:
    \begin{align}
        \mathcal{E} = \{q_z, \ket\Phi_z\},
    \end{align}
    where 
    \begin{align}
        &\ket{\Phi_z} = \ket{\tilde{\Phi}_z}/\sqrt{q_z},\\
        &\ket{\tilde{\Phi}_z} = \left(\mathbb{I}_A\otimes \bra{z}_B \right)\;e^{-iGt}\ket\phi_{AB},\\
        &q_z = \braket{\tilde{\Phi}_z},\\
        &z\in\{0,1\}^{n_B}.
    \end{align}
    Then the ensemble $\mathcal{E}$ is an exact quantum state design in the thermodynamic limit,
    \begin{align}
        \lim_{N_B\rightarrow \infty} \rho^{(k)}_{\mathcal{E}}=\rho^{(k)}_{Haar} \;\;\;\; \forall \,k\in\mathbb{R}_+,
    \end{align}
    for any time $t$ that satisfies $J_1(2t)/t = 0$.
\end{theorem}
\begin{proof}
    We would like to show that at these times $t$, the frame potential for our projected Gaussian ensemble is exactly that of Haar in the thermodynamic limit of the measurement subsystem for every instance of $G$. Instead of calculating the frame potential directly, we will calculate
    \begin{align}
        \lim_{N_B\rightarrow\infty}\Delta^{(k)}_{rms} = \lim_{N_B\rightarrow\infty}\sqrt{\mathbb{E}_{G}\left[\left(\Delta^{(k)}\right)^2\right]} 
    \end{align}
    Since $\Delta^{(k)}_{rms}\geq \mathbb{E}_{G}\Delta^{(k)},$ if $\Delta^{(k)}_{rms}=0$ then $\mathbb{E}_{G}\Delta^{(k)}=0$. Moreover, by Markov's inequality, if $\Delta^{(k)}_{rms}=0$ then the fluctuation also goes to zero, so every instance of $G$ will give an exact design. $\lim_{N_B\rightarrow\infty}\Delta^{(k)}_{rms}=0$ occurs if and only if
    $$\lim_{N_B\rightarrow\infty}\mathbb{E}_{G}F^{(k)} = F^{(k)}_{Haar}.$$
    We will now show this is true.
    The frame potential for the projected ensemble can be written as: 
    \begin{align}
        \mathbb{E}_{G}F^{(k)} = \mathbb{E}_{G}\sum_{z_1, z_2} q_{z_1} q_{z_2} \abs{\langle\Phi_{z_1}|\Phi_{z_2}\rangle}^{2k}.
    \end{align}
    We substitute the unnormalized state appropriately and rewrite this using the replica trick, taking $n=1-k.$ 
    \begin{align}
        \mathbb{E}_{G}F^{(k)} &= \mathbb{E}_{G}\sum_{z_1, z_2} \abs{\left\langle\tilde{\Phi}_{z_1}|\tilde{\Phi}_{z_2}\right\rangle}^{2k} \braket{\tilde{\Phi}_{z_1}}{\tilde{\Phi}_{z_1}}^{1-k}  \braket{\tilde{\Phi}_{z_2}}{\tilde{\Phi}_{z_2}}^{1-k}\\
        &= \mathbb{E}_{G}\sum_{z_1, z_2} \abs{\left\langle\tilde{\Phi}_{z_1}|\tilde{\Phi}_{z_2}\right\rangle}^{2k}  \braket{\tilde{\Phi}_{z_1}}{\tilde{\Phi}_{z_1}}^n  \braket{\tilde{\Phi}_{z_2}}\tilde{\Phi}_{z_2}^n\\
        &= \mathbb{E}_{G}\tr\left(\left(e^{-iGt}|\phi\rangle\langle\phi|e^{iGt}\right)^{\otimes R}_{AB} \mathcal{Q}_B\right),
    \end{align}
    where $R=2(n+k)$ and
    \begin{align}
        \mathcal{Q}_B = \sum_{z_1,z_2} \left|z_1^{\otimes n}z_1^{\otimes k}z_2^{\otimes k}z_2^{\otimes n}\right\rangle \left\langle z_1^{\otimes n}z_2^{\otimes k}z_1^{\otimes k}z_2^{\otimes n}\right|.
    \end{align}
    Let us analyze the state $\left(e^{-iGt}|\phi\rangle\langle\phi|e^{iGt}\right)^{\otimes R}_{AB}$. We can write $e^{-iGt} = UDU^\dagger$ due to the invariance of the GUE under Haar conjugation, and we break up the expectation value to average over $U$ and $D.$
    We can then use the Weingarten calculus to compute. Recall that the action of the Haar twirl can be written as 
    \begin{align}
        \mathcal{N}(\rho)&:=\mathbb{E}_{U}U^{\otimes m}\rho(U^\dagger)^{\otimes m} \notag\\ &= \sum_{\sigma,\tau\in S_{m}} \tr(\rho\sigma^{-1})Wg(\tau\sigma^{-1},N)\tau\\
        &\rightarrow \frac{1}{N^m}\sum_{\sigma\in S_{m}} \tr(\rho\sigma^{-1})\sigma.
    \end{align}
    Here the permutations $\sigma$ and $\tau$ act on the $k$ copies; if $\rho\in\mathcal{H}^{\otimes m}_{AB}$, the $\sigma$ and $\tau$ act on $A$ and $B$ each separately. $Wg$ is the Weingarten function which is some rational function in the dimension $N$ of thhe Hilbert space on which $U$ acts. The last approximation comes from Lemma 1 proved in \cite{schuster2025random}; applied in the limit $N\rightarrow\infty$, the expression is exact for the Haar twirl.

\begin{figure}
    \centering
\begin{tikzpicture}[scale=0.5]
    
    \draw[thick] (-7.5,-0.5) -- (-7.5,0.5) -- (-6.5,0.5) -- (-6.5,-0.5) -- (-7.5,-0.5);
    \draw[thick] (-5.5,-0.5) -- (-5.5,0.5) -- (-4.5,0.5) -- (-4.5,-0.5) -- (-5.5,-0.5);
    \draw[thick] (-3.5,-0.5) -- (-3.5,0.5) -- (-2.5,0.5) -- (-2.5,-0.5) -- (-3.5,-0.5);
    \draw[thick] (-1.5,-0.5) -- (-1.5,0.5) -- (-0.5,0.5) -- (-0.5,-0.5) -- (-1.5,-0.5);
    \draw[thick] (1.5,-0.5) -- (1.5,0.5) -- (0.5,0.5) -- (0.5,-0.5) -- (1.5,-0.5);
    \draw[thick] (3.5,-0.5) -- (3.5,0.5) -- (2.5,0.5) -- (2.5,-0.5) -- (3.5,-0.5);
    \draw[thick] (5.5,-0.5) -- (5.5,0.5) -- (4.5,0.5) -- (4.5,-0.5) -- (5.5,-0.5);
    \draw[thick] (7.5,-0.5) -- (7.5,0.5) -- (6.5,0.5) -- (6.5,-0.5) -- (7.5,-0.5);
    \node at (-7,0) {$\bm{U}$};
    \node at (-5,0) {$\bm{U}^\dagger$};    
    \node at (-3,0) {$\bm{U}$};
    \node at (-1,0) {$\bm{U}^\dagger$};
    \node at (1,0) {$\bm{U}$};
    \node at (3,0) {$\bm{U}^\dagger$};
    \node at (5,0) {$\bm{U}$};
    \node at (7,0) {$\bm{U}^\dagger$};
    
    \draw[thick] (-3,-0.5) -- (-3,-1);
    \draw[thick] (-3,0.5) -- (-3,1);
    \draw[thick] (3,-0.5) -- (3,-1);
    \draw[thick] (3,0.5) -- (3,1);
    \draw[thick] (-1,-0.5) -- (-1,-1);
    \draw[thick] (-1,0.5) -- (-1,1);
    \draw[thick] (1,-0.5) -- (1,-1);
    \draw[thick] (1,0.5) -- (1,1);
    \draw[thick] (-5,-0.5) -- (-5,-1);
    \draw[thick] (-5,0.5) -- (-5,1);
    \draw[thick] (5,-0.5) -- (5,-1);
    \draw[thick] (5,0.5) -- (5,1);
    \draw[thick] (-7,-0.5) -- (-7,-1);
    \draw[thick] (-7,0.5) -- (-7,1);
    \draw[thick] (7,-0.5) -- (7,-1);
    \draw[thick] (7,0.5) -- (7,1);

    \draw[thick] (-7,-1) arc (180:240:1);
    \draw[thick] (-5,-1) arc (0:-60:1);
    \draw[thick] (-5,1) arc (180:150:1);
    \draw[thick] (-3,1) arc (0:30:1);
    \draw[thick] (-3,-1) arc (180:240:1);
    \draw[thick] (-1,-1) arc (0:-60:1);
    \draw[thick] (1,-1) arc (180:240:1);
    \draw[thick] (3,-1) arc (0:-60:1);
    \draw[thick] (3,1) arc (180:150:1);
    \draw[thick] (5,1) arc (0:30:1);
    \draw[thick] (5,-1) arc (180:240:1);
    \draw[thick] (7,-1) arc (0:-60:1);

    \node at (-6,-2) {$\bm{D}$};
    \node at (-4,2) {$\ketbra{\phi}$};
    \node at (-2,-2) {$\bm{D}^\dagger$};
    \node at (2,-2) {$\bm{D}$};
    \node at (4,2) {$\ketbra{\phi}$};
    \node at (6,-2) {$\bm{D}^\dagger$};

    \node at (9,0) {...};

    

    \end{tikzpicture}%
    \caption{A tensor contraction or Weingarten diagram for $(UDU^\dagger\ketbra{\phi} UD^\dagger U^\dagger)^{\otimes R}$. Notably the bottom factors of $D$ and $D^\dagger$ are all contracted together via $\sigma\in S_{2R}$, and the top has empty legs where the $\mathcal{Q}_B$ and overall trace will come in.}\label{fig:wg_diagram}
    \end{figure}

    For our setup, we can let $m=2R$ and choose $\rho$ such that it alternates between $D$ and $D^\dagger$ between copies. We can use the infinite limit since we have stipulated $N_B\rightarrow\infty$. Then we can write
    \begin{align}
        &\mathbb{E}_{U,D}(UDU^\dagger\ketbra{\phi} UD^\dagger U^\dagger)^{\otimes R} \notag\\ &= \mathbb{E}_{D}\mathcal{N}(\rho)(\ketbra{\phi}^{\otimes R})\\
        &=\mathbb{E}_{D}\frac{1}{N_{AB}^{2R}}\sum_{\sigma\in S_{2R}} \tr(\rho_{D,D^\dagger} \sigma) C(\sigma, \ketbra{\phi}^{\otimes R}),
    \end{align}
    where the action of $\mathcal{N}(\rho)$ and $C$ are consistent and appropriate sets of contractions of $\ketbra{\phi}^{\otimes R}$. See \autoref{fig:wg_diagram} for a Weingarten diagram. Recall that this whole object ultimately outputs a density matrix in $\mathcal{H}^{\otimes R}_{AB}$, which is consistent with $\sigma\in S_{2R}$ contracting upon an $R$-copy density matrix. 
    Then we find
    \begin{align}
        \mathbb{E}_{G}F^{(k)} &= \frac{1}{N_{AB}^{2R}}\sum_{\sigma\in S_{2R}} \mathbb{E}_{D}\tr(\rho_{D,D^\dagger} \sigma) \tr_{AB}\left(C(\sigma, \ketbra{\phi}^{\otimes R})_{AB}\mathcal{Q}_{B}\right).
    \end{align}
    Our goal now is to show that only a very select number of permutations $\sigma$ contribute in the $N_B\rightarrow\infty$ limit; these permutations will give us exactly the Haar frame potential. We begin by noting some facts about the first term in the above expression, $\mathbb{E}_{D}\tr(\rho_{D,D^\dagger} \sigma)$:
    \begin{itemize}
        \item For each permutation $\sigma$, this term is a product of traces of products of $D$ and $D^\dagger$, 
        $$\text{e.g.}\;\;\; \mathbb{E}_{D}\tr(D^\dagger)\tr(D^\dagger)\tr(D^2),\mathbb{E}_{D}\tr(DD^\dagger)\tr(D)\tr(D^\dagger), \text{etc}.$$
        Every one of the $R$ copies of $D$'s and $R$ copies of $D^\dagger$ is accounted for in some trace term.
        \item From \lem{good_t_exists}, we know that at the points $t$ specified, as $N_B\rightarrow\infty$, $$\frac{1}{N}\mathbb{E}_{D}\tr(D)=\frac{1}{N}\mathbb{E}_{D}\tr(D^\dagger)=0.$$
        From \lem{single_exp_gue_concentration}, we find as $N_B\rightarrow\infty$ the concentration becomes exact. Together this means that for all instances of $G,$
        $$\frac{1}{N}\tr(D)=\frac{1}{N}\tr(D^\dagger)=0.$$
        As such, any term with $\tr(D)$ or $\tr(D^\dagger)$ will be zero.
        \item This implies that every nonzero trace term must have at least two matrices (2 copies of $D$, 2 copies of $D^\dagger$, or one of each). This means at most there are $R$ trace terms. Hence the contribution of this first term cannot be greater than $N_{AB}^R$.
        \item Thus, any permutation $\sigma$ that does not have an $N_{AB}^R$ contribution from $\tr_{AB}\left(C(\sigma, \ketbra{\phi}^{\otimes R})_{AB}\mathcal{Q}_{B}\right)$ will vanish as $N_B\rightarrow\infty$.
    \end{itemize}
    This last point is our condition for eliminating permutations $\sigma$. It seems demanding, but recall that at some point we will plug in $R=2.$ Thus our final expression will only sum over permutations $\sigma$ that give us $N_B^2$ or more from that second term. We now turn to analyzing the second term $\tr_{AB}\left(C(\sigma, \ketbra{\phi}^{\otimes R})_{AB}\mathcal{Q}_{B}\right)$:
    \begin{itemize}
        \item This term is a series of contractions depending on the permutation $\sigma$, and hence, they are products of trace terms of $\ketbra{\phi}_{AB}$ factors and $\mathbb{I}_A\otimes|z_i\rangle\langle z_j|_B$ factors, where $i=1,2$, $j=1,2$
        \begin{align*}
            \text{e.g.} \;\;\; &\tr\left(\ketbra{\phi}(\mathbb{I}\otimes|z_2\rangle\langle z_1|)\right)\tr\left(\ketbra{\phi}(\mathbb{I}\otimes|z_1\rangle\langle z_2|)\right),\\
            &\tr(\ketbra{\phi}^2)\tr(\mathbb{I}\otimes|z_2\rangle\langle z_2|)\tr(\mathbb{I}\otimes|z_1\rangle\langle z_1|)\;\;\;\text{etc}.
        \end{align*}
        Every one of the $R$ copies of $\ketbra{\phi}$ and $R$ copies of $\mathbb{I}\otimes|z_i\rangle\langle z_j|$ are accounted for in some trace. 
        \item There are permutations $\sigma$ that will cause all the $z$'s and all the $\ketbra{\phi}$ to be in separate traces. A trace term that consists only of $\ketbra{\phi}$ factors will evaluate to 1. A trace term that consists only of $\mathbb{I}\otimes|z_1\rangle\langle z_1|$ will evaluate to $N_A$ (and same for $z_2$). Any term with both $z_1$ and $z_2$ after simplification will vanish.
        \item There are also permutations $\sigma$ that will cause traces to mix $z$'s and $\ketbra{\phi}$. In particular, those will have terms involving $\ket{\phi_{z_1}}=\langle z_1|\phi\rangle \in\mathcal{H}_A$ or $\ket{\phi_{z_2}}=\langle z_2|\phi\rangle\in\mathcal{H}_A$ or their conjugate. Note these states are not normalized.
        
        If a permutation $\sigma$ creates a $\ket{\phi_{z_1}}$ factor, then it must also create its conjugate (and the same for $z_2$). This is because every nonzero $\bra{z_1}$ must multiply onto $\ket{z_1}$ or $\ket{\phi}$; it vanishes if it multiplies on $\ket{z_2}$. But there are the same number ($R/2$) of $\bra{z_1}$ as there are $\ket{z_1}$, so any $\bra{z_1}$ that multiplies with $\ket{\phi}$ also leaves a $\ket{z_1}$ to multiply a $\bra{\phi}.$ The factor and its conjugate, however, may be involved in different traces.

        Then there are four possible types of product terms after simplification: $\braket{\phi_{z_1}}$, $\braket{\phi_{z_2}}$, $\langle\phi_{z_2}|\phi_{z_1}\rangle$ and its conjugate. The first two evaluate to $p_{z_1}$ (and $p_{z_2}$) which are the probabilities of getting $z_1$ (and $z_2$) from $\ket\phi$. The second two always appear with each other.
        \item The $\mathcal{Q}_B$ holds a sum over $z_1,z_2$ that has $N_{B}^2$ terms total. 
    \end{itemize}
    Now we analyze the whole term including the sum in $\mathcal{Q}_B$. There are two cases: $z_1=z_2$ and $z_1\neq z_2$. 
    
    In the case where $z_1=z_2$, the sum has $N_{B}$ terms of at most order 1 in $N_B$, and the sum has only $N_{B}$ terms so taking into account the rest of the frame potential, it does not contribute enough to survive the $N_B\rightarrow\infty$ limit. 
    
    We can then focus on the case where $z_1\neq z_2$. This in turn has two types of permutations: one where the $z$'s and the $\ketbra\phi$ are in separate traces, and one where they are mixed.

    Consider the permutations that have traces that mix $z$'s and $\ketbra{\phi}$. All the unmixed traces only contribute factors of 1 and $N_A$, so we need to see if the mixed traces give any factors of $N_B$. We notice $\sum_{z_1}p_{z_1}^l \leq 1$ for $l\geq 1$, and similarly for $z_2$, so only the mixed traces of $\langle\phi_{z_2}|\phi_{z_1}\rangle$ (appearing with its conjugate), are left. But since it appears with its conjugate, the state overlaps $\langle\phi_{z_2}|\phi_{z_1}\rangle\langle\phi_{z_1}|\phi_{z_2}\rangle$ have a factor of $p_{z_1}p_{z_2}$ as well, and so
    $$\sum_{z_1,z_2}\langle\phi_{z_2}|\phi_{z_1}\rangle^l\langle\phi_{z_1}|\phi_{z_2}\rangle^l \leq\sum_{z_1,z_2} p_{z_1}^lp_{z_2}^l\leq 1.$$
    Hence the permutations that mix traces vanish as $N_B\rightarrow\infty.$

    This leaves the set of permutations where the $z$'s and the $\ketbra\phi$ are in separate traces. That means that essentially the permutation $\sigma$ can be broken down into the permutation that affects $z$'s and the permutation that affects $\ketbra\phi$. We will call the latter $\sigma'$, and for all $\sigma'$ we note the contribution from $\ketbra\phi$ will be 1. For the former, notice that for the trace with $\mathcal{Q}_B$ to not vanish since $z_1\neq z_2$, we need the central $k$ copies of $z_i$ to be swapped. This means that for the term to not vanish, the permutation that affects the $z$'s must take a form $\tau\pi_1\pi_2$ where $\tau$ is a transposition of the center $k$ copies, and $\pi_1$ (and $\pi_2$) permute the first $n+k$ (and second $n+k$) copies.
\begin{align}
        \mathbb{E}_{G}F^{(k)} &= \frac{1}{N_{AB}^{2R}}\sum_{\sigma'\in S_{R}, \pi_1,\pi_2\in S_{R/2}} \mathbb{E}_{D}\tr(\rho_{D,D^\dagger} \sigma'\tau\pi_1\pi_2) \sum_{z_1,z_2}\tr_{AB}\left(\tau\pi_1\pi_2\mathcal{Q}_B\right)\\
        &= \frac{N_B^2-N_B}{N_{AB}^{2R}}\sum_{\sigma'\in S_{R}, \pi_1,\pi_2\in S_{R/2}} \mathbb{E}_{D}\tr(\rho_{D,D^\dagger} \sigma'\tau\pi_1\pi_2) \tr_{A}\left(\tau\pi_1\pi_2\right).
    \end{align}
    
    Now we turn to $\mathbb{E}_{D}\tr(\rho_{D,D^\dagger} \sigma'\tau\pi_1\pi_2)$. These permutations do a very particular thing to this term: look at \autoref{fig:wg_diagram}. To keep the $z$ factors together (which are the empty top legs in \autoref{fig:wg_diagram}), $\tau\pi_1\pi_2$ must contract an ingoing $D$ leg with an outgoing $D^\dagger$ leg. Each choice of $\tau\pi_1\pi_2$ gives a very particular pattern in which each ingoing $D$ leg contracts with each outgoing $D^\dagger$ legs. Moreover, to keep $\ketbra\phi$ connected to each other, $\sigma'$ must contract an ingoing $D^\dagger$ leg with an outgoing $D$ leg. Then for each choice of $\tau\pi_1\pi_2$, there is exactly one choice of $\sigma'$ which contracts the same $D$ and $D^\dagger$ together such that
    \begin{align}
        \mathbb{E}_{D}\tr(\rho_{D,D^\dagger} \sigma'\tau\pi_1\pi_2) = \mathbb{E}_{D}\tr(DD^\dagger)^R = \tr(\mathbb{I}_{AB})^R = N_{AB}^R.
    \end{align}
    Here we recall that $D$ comes from a unitary. Every other $\sigma'$ creates factors of $\tr((DD^\dagger)^p)$ for $p\geq 1$ which means $\mathbb{E}_{D}\tr(\rho_{D,D^\dagger} \sigma'\tau\pi_1\pi_2)\leq \mathcal{O}(N^{R})$ and so they will not survive as $N_B\rightarrow\infty$. At last we have the set of permutations that do not vanish:
    \begin{align}
        \mathbb{E}_{G}F^{(k)} &= \frac{N_B^2-N_B}{N_{AB}^{R}}\sum_{\pi_1,\pi_2\in S_{R/2}} \tr_{A}\left(\tau\pi_1\pi_2\right).
    \end{align}
    Now we show that this gives the Haar frame potential:
    \begin{align}
        \mathbb{E}_{G}F^{(k)} &= \frac{N_B^2-N_B}{N_{AB}^{R}}\sum_{\pi_1,\pi_2\in S_{R/2}} \tr_{A}\left(\tau\pi_1\pi_2\right)\\
        &=\frac{N_B^2-N_B}{N_{AB}^{R}}\left(\frac{(R/2+N_A-1)!}{(N_A-1)!}\right)^2 {}\tr_{A}\left(\tau\left(\rho^{(n+k)}_{Haar}\otimes\rho^{(n+k)}_{Haar}\right)\right)\\
        &=\frac{N_B^2-N_B}{N_{AB}^{R}}\left(\frac{(R/2+N_A-1)!}{(N_A-1)!}\right)^2F^{(k)}_{Haar}\\
        &\rightarrow F^{(k)}_{Haar}.
    \end{align}
    In the second line, we've identified the permutations as a representation of $k$-th moment Haar average density matrix and substituted appropriately. In the third line, we've traced over $n$, which are not affected by $\tau$, and used $\tr(\tau \rho^{\otimes 2})=\tr\left(\rho^2\right)$. In the last line we've substituted for $R=2$ and taken the infinite $N_B$ limit.
\end{proof}
At the early set of discrete times $t$ that satisfy the above theorem, \lem{singleG_moments} tells us that the global system will only be an approximate $\mathcal{O}(1)$ design. This means that at those early times, to the best of our understanding, in the thermodynamic limit of the measured system, we are going from a very low degree approximate design (exact degree depends on exact $t$) to an exact $\infty$-design with the tools of the projected ensemble. We have boosted a lot of design randomness out of our chaotic system!


\section{Generalizations of the one Gaussian Hamiltonian projected ensemble}

Our proof techniques for the GUE Hamiltonian projected ensemble generalize to a variety of modifications. The first modification we will look at is a spectral modification: ensembles that are invariant under Haar unitary conjugations and have points where the first moment trace is zero. 
\begin{corollary}\label{cor:singleG_spectral_mod}
    Define a random unitary $V = UDU^{\dagger}$ where $V\in\mathcal{U}(N_{AB})$, random matrix $U\in\mathcal{U}(N_{AB})$ is Haar distributed, and random matrix $D\in\mathcal{U}(N_{AB})$ satisfies:
    \begin{enumerate}
        \item $\tr D = 0$,
        \item $\operatorname{Pr}[\left|\tr D - \E\tr D\right| \geq s ] \rightarrow0$ as $N_B\rightarrow\infty$. 
    \end{enumerate} Then consider the projected ensemble formed from this unitary evolution on some initial state $\ket\phi$ and measurement in the computational basis on $n_B$ qubits:
    \begin{align}
        \mathcal{E} = \{q_z, \ket\Phi_z\},
    \end{align}
    where 
    \begin{align}
        &\ket{\Phi_z} = \ket{\tilde{\Phi}_z}/\sqrt{q_z},\\
        &\ket{\tilde{\Phi}_z} = \left(\mathbb{I}_A\otimes \bra{z}_B \right)\;V\ket\phi_{AB},\\
        &q_z = \braket{\tilde{\Phi}_z},\\
        &z\in\{0,1\}^{n_B}.
    \end{align}
    Then the ensemble $\mathcal{E}$ is an exact quantum state design in the thermodynamic limit,
    \begin{align}
        \lim_{N_B\rightarrow \infty} \rho^{(k)}_{\mathcal{E}}=\rho^{(k)}_{Haar} \;\;\;\; \forall \,k\in\mathbb{R}_+.
    \end{align}
\end{corollary}
\begin{proof}
    The proof for \autoref{thm:projected_designs_GUE} above can be generalized by substituting for $D$ in $\mathbb{E}_{D}\tr(\rho_{D,D^\dagger} \sigma)$ appropriately.
\end{proof}
A special case of this above scenario is when the spectrum is actually a fixed matrix $D$ and the random ensemble is generated solely through the Haar randomness of the eigenbases. We also, however, can make an eigenbasis modification:
\begin{corollary}\label{cor:singleG_basis_mod}
    Define a random unitary $V = UDU^{\dagger}$ where $V\in\mathcal{U}(N_{AB})$, random matrix $U\in\mathcal{U}(N_{AB})$ is an approximate 4-design for an exponentially small additive error $\epsilon/N_{A}^{k+4}$, and random matrix $D\in\mathcal{U}(N_{AB})$ satisfies:
    \begin{enumerate}
        \item $\tr D = 0$,
        \item $\operatorname{Pr}[\left|\tr D - \E\tr D\right| \geq s ] \rightarrow0$ as $N_B\rightarrow\infty$. 
    \end{enumerate} Then consider the projected ensemble formed from this unitary evolution on some initial state $\ket\phi$ and measurement in the computational basis on $n_B$ qubits:
    \begin{align}
        \mathcal{E} = \{q_z, \ket\Phi_z\},
    \end{align}
    where 
    \begin{align}
        &\ket{\Phi_z} = \ket{\tilde{\Phi}_z}/\sqrt{q_z},\\
        &\ket{\tilde{\Phi}_z} = \left(\mathbb{I}_A\otimes \bra{z}_B \right)\;V\ket\phi_{AB},\\
        &q_z = \braket{\tilde{\Phi}_z},\\
        &z\in\{0,1\}^{n_B}.
    \end{align}
    Then the ensemble $\mathcal{E}$ is an approximate quantum state $k$-design with $1-\mathcal{O}(1/N_A)$ probability in the thermodynamic limit $N_B\rightarrow\infty$, for additive error $\epsilon'=\mathcal{O}(N_A)\sqrt{\frac{\epsilon}{N_A^{k+4}} {{k+N_A-1}\choose{k}}}\leq\mathcal{O}(1/N_A)$. 
\end{corollary}
\begin{proof}
    We work again with the frame potential:
    \begin{align}
        \lim_{N_B\rightarrow\infty}\mathbb{E}_{U,D}F^{(k)}  &= \lim_{N_B\rightarrow\infty}\mathbb{E}_{U,D}\sum_{z_1, z_2} q_{z_1} q_{z_2} \abs{\langle\Phi_{z_1}|\Phi_{z_2}\rangle}^{2k}\\
        &=\lim_{N_B\rightarrow\infty}\mathbb{E}_{U,D}\tr\left(\left(UDU^\dagger|\phi\rangle\langle\phi|UD^\dagger U^\dagger\right)^{\otimes R}_{AB} \mathcal{Q}_B\right).
    \end{align}
    Define the ensemble $W$ to be Haar random unitaries. By H\"older's inequality on Schatten p-norms, we find
\begin{align}\notag
        &\abs{\lim_{N_B\rightarrow\infty}\mathbb{E}_{U,D}F^{(k)}-F^{(k)}_{Haar}}  \\       &=\lim_{N_B\rightarrow\infty}\abs{\tr\left(\left[\mathbb{E}_{U,D}\left(UDU^\dagger|\phi\rangle\langle\phi|UD^\dagger U^\dagger\right)^{\otimes R}-\mathbb{E}_{W,D}\left(WDW^\dagger|\phi\rangle\langle\phi|WD^\dagger W^\dagger\right)^{\otimes R}\right] \mathcal{Q}_B\right)}\\
        &\leq \lim_{N_B\rightarrow\infty}\abs{\mathcal{Q}_B}_{\infty}\abs{\mathbb{E}_{U,D}\left(UDU^\dagger|\phi\rangle\langle\phi|UD^\dagger U^\dagger\right)^{\otimes R}-\mathbb{E}_{W,D}\left(WDW^\dagger|\phi\rangle\langle\phi|WD^\dagger W^\dagger\right)^{\otimes R}}_1\\
        &\leq \lim_{N_B\rightarrow\infty}\abs{\mathbb{E}_{U,D}\left(UDU^\dagger|\phi\rangle\langle\phi|UD^\dagger U^\dagger\right)^{\otimes R}-\mathbb{E}_{W,D}\left(WDW^\dagger|\phi\rangle\langle\phi|WD^\dagger W^\dagger\right)^{\otimes R}}_1
    \end{align}
    Recall from our proof that $R=2$. From the lemma statement $U$ is (at least) an approximate 4-design with exponentially small additive error $\epsilon/N^k_{A}$. This implies that
    \begin{align}
        \abs{\E_{U} U^{(4)}(\cdot)(U^\dagger  )^{(4)}-  \E_{W}W^{(4)}(\cdot)(W^\dagger  )^{(4)}}_{\diamond}\leq \frac{\epsilon}{N_A^{k+4}}.
    \end{align}
    Our difference of frame potentials are an explicit example of these channels then. Since $F^{(k)}\geq F^{(k)}_{Haar}$, we find
    \begin{align}
        \lim_{N_B\rightarrow\infty}\mathbb{E}_{U,D}F^{(k)} \leq F^{(k)}_{Haar}+\frac{\epsilon}{N_A^{k+4}}
    \end{align}
    This means that
    \begin{align}
        \lim_{N_B\rightarrow\infty} \Delta^{(k)}_{rms}= \lim_{N_B\rightarrow\infty}\sqrt{\mathbb{E}_{U,D}\frac{F^{(k)}}{F^{(k)}_{Haar}}-1}\leq \sqrt{\frac{\epsilon}{N_A^{k+4}F^{(k)}_{Haar}}} = \sqrt{\frac{\epsilon}{N_A^{k+4}} {{k+N_A-1}\choose{k}}}.
    \end{align}
    Notice that this is $\mathcal{O}(1/N_A^2)$.
    We then recall that 
    \begin{align}
        \E \norm{\rho^{(k)} - \rho^{(k)}_{Haar}}_1 \leq \E \Delta^{(k)}\leq \Delta^{(k)}_{rms}.
    \end{align}
    Then we see that the expectation of the one norm is controlled by the $\Delta^{(k)}_{rms}$, and we apply Markov's inequality to get the concentration probability.
\end{proof}
Let us seek to understand this result. First let us address the assumptions: though the error required for the basis scales with dimension, it scales with the dimension of the system left after measuring, $N_A$. This system could be much smaller than the global size. Even if it isn't, though, the design requirement still only applies to the basis, not the entire system. As a matter of fact, the entire system could badly fail to be a design depending on the characteristics of the spectrum---so again, given the right error on just the eigenbasis and specific conditions on the spectral nature, we will have boosted a low degree design to a higher degree $k$-design with the tools of the projected ensemble.

We are not currently aware of specific chaotic systems that satisfy these requirements, but this is mostly because the eigenbases of random Hamiltonian ensembles have not been previously studied. We leave this for future work. An approximate 4-design basis requirement, though, is a much more relaxed requirement than the Haar eigenbases requirements before.


\section{Finite $N$ for the one Gaussian Hamiltonian projected ensemble}

Now we discuss the robustness of GUE Hamiltonian projected ensemble as we back away from the thermodynamic limit and return to finite system size. Away from the thermodynamic limit, we expect a number of finite $N_{B}$ corrections that will come from a number of the permutations we eliminated previously in the proof of \autoref{thm:projected_designs_GUE}. We will now account for those.

\begin{lemma}\label{lem:finite_N_corrections_GUE}
    Let $G$ be an $N_{AB}$-dimensional random matrix drawn from the GUE that acts on $\mathcal{H}_{A}\otimes\mathcal{H}_B$. Then consider the projected ensemble formed from time evolution under this Gaussian Hamiltonian on some initial state $\ket\phi$ and measurement in the computational basis on $n_B$ qubits:
    \begin{align}
        \mathcal{E} = \{q_z, \ket\Phi_z\},
    \end{align}
    where 
    \begin{align}
        &\ket{\Phi_z} = \ket{\tilde{\Phi}_z}/\sqrt{q_z},\\
        &\ket{\tilde{\Phi}_z} = \left(\mathbb{I}_A\otimes \bra{z}_B \right)\;e^{-iGt}\ket\phi_{AB},\\
        &q_z = \braket{\tilde{\Phi}_z},\\
        &z\in\{0,1\}^{n_B}.
    \end{align}
    Then the projected ensemble $\mathcal{E}$ satisfies
    \begin{align}
        \E_G F^{(k)} \leq  F^{(k)}_{Haar}+ \mathcal{O}(t^2/N_A^2N_B).
    \end{align}
    and with a probability of at least $1-\mathcal{O}(1/N_A)$, it is a $\epsilon$-approximate state $k$-design for any $\epsilon < \mathcal{O}(1/N_A)$, for any $k$ such that $N_B>N_A^{k+3}$, for time $t < N_A$ that satisfies $J_1(2t)/t = 0$, and for $N_{AB}>4^{16}$. 
\end{lemma}
\begin{proof}
    Recall the definition that
    \begin{align}
        \Delta^{(k)}_{rms} =\sqrt{\mathbb{E}_{G}\left[\left(\Delta^{(k)}\right)^2\right]} = \sqrt{\frac{\E F^{(k)}}{F^{(k)}_{Haar}}-1}.
    \end{align}
    We proceed by first finding the finite $N$ corrections to $\frac{\E F^{(k)}}{F^{(k)}_{Haar}}$. Recall from our proof of \autoref{thm:projected_designs_GUE} that once we write out the frame potential and use the replica trick, we find that
\begin{align}\label{eq:finite_n_eq_breakdown}
        \E_{G}F^{(k)} &= \mathbb{E}_{G}\tr\left(\left(e^{-iGt}|\phi\rangle\langle\phi|e^{iGt}\right)^{\otimes R}_{AB} \mathcal{Q}_B\right)\\
        &= \E_D \sum_{\sigma,\tau\in S_{2R}} \tr(\rho_{D,D^{\dagger}}\sigma^{-1})Wg(\tau\sigma^{-1},N_{AB})\tr\left(C(\tau, \ketbra{\phi}^{\otimes R})_{AB}\mathcal{Q}_{B})\right).
    \end{align}
    for $Wg$ as the Weingarten functions, $C$ as some function dependent on the Weingarten contractions, $R=2(n+k)=2$, and 
    \begin{align}
        \mathcal{Q}_B = \sum_{z_1,z_2} \left|z_1^{\otimes n}z_1^{\otimes k}z_2^{\otimes k}z_2^{\otimes n}\right\rangle \left\langle z_1^{\otimes n}z_2^{\otimes k}z_1^{\otimes k}z_2^{\otimes n}\right|.
    \end{align}
    This is a sum over permutation pairs $(\sigma, \tau) \in S_{2R}.$ Let us split the set of permutations into two: the pairs of the form $(\sigma,\sigma)$, $\sigma=\sigma'\tau\pi_1\pi_2$ that contributed nonzero terms in the infinite $N$ case (in other words, all the ones that appeared in the proof of \autoref{thm:projected_designs_GUE}) and those that did not.

    \textit{Case 1:} For the permutation pairs that contributed nonzero infinite $N$ terms, the $\tr(\rho_{D,D^{\dagger}})$ term is the same (and gives the expected Haar frame potential):
    \begin{align}
        \mathbb{E}_{D}\tr(\rho_{D,D^\dagger} \sigma'\tau\pi_1\pi_2) = N_{AB}^R.
    \end{align}
    This is exact in the finite $N$ limit since the structure of the permutations means the trace does not actually depend on $D$ or $D^\dagger$ coming from GUE. However, now there will be corrections to the Weingarten function. From Lemma E.1 in \cite{chen2024efficient}, we have that, for a permutation $\pi\in S_q$ and $N\geq \Omega(q^{4q})$, the asymptotics of the Weingarten function can be controlled by:
    \begin{align}
        Wg(\pi) = N^{-2q+cycles(\pi)}\prod_i (-1)^{|C_i|-1}\text{Cat}_{|C_i|-1}  +\tilde{\mathcal{O}}(N^{-2q+cycles(\pi)-3/2})
    \end{align}
    where $C_i$ are the lengths of the $i$th cycle in $\pi$ and $\text{Cat}_n$ are the Catalan numbers. The permutation pairs that contributed nonzero infinite $N$ terms have $\pi=\sigma\tau^{-1}=(1)^{2R}$. Then taking $q=2R,$ we see that these permutations have $2R$ cycles, so the corrections will be $\tilde{\mathcal{O}}(N^{-2R-3/2})$ for $N=N_{AB}$:
    \begin{align}
        Wg(\tau\pi_1\pi_2) = N^{-2R}+\tilde{\mathcal{O}}(N^{-2R-3/2})
    \end{align}
    Lastly, we must now consider $\mathcal{Q}_B$ term. Previously we only considered the situation where $z_1\neq z_2$. Now recall that for the case of $z_1=z_2$, the sum will have at most $N_{B}$ terms of the form
    \begin{align}
        \tr(C(\tau, \ketbra{\phi}^{\otimes R})_{AB}\ketbra{z^{\otimes R}}) < 1.
    \end{align}
    This means that in total, for each $\tau\pi_1\pi_2$
    \begin{align}
        \tr\left(C(\tau, \ketbra{\phi}^{\otimes R})_{AB}\mathcal{Q}_{B})\right) = \sum_{z_1,z_2}\tr_{AB}(\tau\pi_1\pi_2) = (N_B^2-N_B)\tr_{A}(\tau\pi_1\pi_2) +\mathcal{O}(N_B).
    \end{align}
    In summary,
\begin{align}\notag
        &\E_D \sum_{\pi_1,\pi_2\in S_{R/2}} \tr(\rho_{D,D^{\dagger}}\sigma^{-1})Wg(\tau\sigma^{-1},N_{AB})\tr\left(C(\tau, \ketbra{\phi}^{\otimes R})_{AB}\mathcal{Q}_{B})\right) \\
        &\leq \sum_{\pi_1,\pi_2\in S_{R/2}} N_{AB}^R (N_{AB}^{-2R}+\tilde{\mathcal{O}}(N_{AB}^{-2R-3/2}))((N_B^2-N_B)\tr_{A}(\tau\pi_1\pi_2) +\mathcal{O}(N_B))
    \end{align}
    Taking the limit as $R=2$, we get
    \begin{align}
        F_{Haar}^{(k)} +\mathcal{O}(1/(N_A^{2}N_B)).
    \end{align}

    \textit{Case 2:} Within the types of permutations $(\sigma,\tau)$ that disappeared in the $N_B\rightarrow\infty$ limit, there are a number of subclasses to consider. First, there are permutation pairs that allow $\tr(\rho_{D,D^{\dagger}}\sigma^{-1})$ to form $R$ separate trace terms in paired traces (e.g. $\tr D^2$, $\tr (D^\dagger)^2$ or $\tr DD^\dagger$). Then the contribution from this term would be $N_{AB}^R$. 
    However, since these are not the permutations that contributed to the Haar answer, either the Weingarten contribution will be smaller because $\sigma\neq\tau$:
    \begin{align}
        Wg(\pi) \leq N^{-2q+q-1}\prod_i (-1)^{|C_i|-1}\text{Cat}_{|C_i|-1} +\tilde{\mathcal{O}}(N^{-2q+q-1-3/2}) = \mathcal{O}(1/N_{AB}^{2R+1})
    \end{align}
    or the permutations will mix the $z$ and $\ketbra{\phi}$ terms in $C$, so 
    \begin{align}
        \tr\left(C(\tau, \ketbra{\phi}^{\otimes R})_{AB}\mathcal{Q}_{B})\right) = \mathcal{O}(N_B).
    \end{align}
    Then the total contribution per permutation would be $\mathcal{O}(1/N_B^{R-1}).$

    Now we look at all the other permutations $(\sigma,\tau)$ that disappear in the $N_B\rightarrow\infty$ limit. These will no longer have the maximum number of trace terms, or in other words, they will not all be paired traces.
    We will deal with them less precisely by estimating what the largest contribution any permutation can make is, then multiply it by $2R!=4!$ to get a loose upper bound on the possible corrections.

    Again we look at the contributions from each of the three possible terms in \autoref{eq:finite_n_eq_breakdown}. First, there will again be the Weingarten function, which can at most be 
    \begin{align}
        Wg(\pi) = N^{-2R}+\tilde{\mathcal{O}}(N^{-2R-3/2})
    \end{align}
    for the case where $\sigma=\tau$ and there are $q=2R$ cycles. Second, there will again be contributions from the $\mathcal{Q}_B$ sum, which are at most $\mathcal{O}(N_B^2).$ 
    
    The final difficulty will be the $\tr(\rho_{D,D^{\dagger}}\sigma^{-1})$ term. 
    Recall that in the $N_B\rightarrow\infty$ limit, we do not include terms that have single traces like $\tr D$ and $\tr D^\dagger$ in the product; the spectral properties of $e^{iGt}$ dictate that in the infinite limit, single trace terms are zero. However, in the finite $N$ limit, we have from \lem{exp_gue_tr_avg} that
    \begin{align}
        \left|\mathbb{E}\frac{1}{N}\tr[e^{i G t}]\right| \leq  \frac{K_0t}{N}.
    \end{align}
    We are now looking at permutation pairs that cannot make all paired trace terms. This means the largest possible contribution will come from one of the following types of trace terms:
    \begin{enumerate}
        \item 4th-power trace term and $(R-4)/2$ paired terms (e.g. $\tr(D^4)\tr((D^\dagger)^2)\tr((D^\dagger)^2)$)
        \item two 3rd-power trace terms and  $(R-6)/2$ paired terms (e.g. $\tr(D^3)\tr((D^\dagger)^3)\tr(DD^\dagger )$)
        \item one single trace term, one 3rd power trace term and $(R-4)/2$ paired terms \\(e.g. $\tr(D)\tr(D(D^\dagger)^2)\tr(DD^\dagger )$)
        \item two single trace terms and paired terms (e.g. $\tr(D)^2\tr((D^\dagger)^2)\tr(DD^\dagger )$).
    \end{enumerate}
    Which set of trace products is the largest? In order to figure this out, we need to both take the expectation value, which will now require some finagling with finite $N$ concentration since the traces do not just evaluate to identity, and also look at $R=2$. In particular, taking $R=2$ shows us that (1) simply gives you a power of $N$, and (2) cannot exist. Any version of (3) will give $|\tr(D)|^2$. Using \lem{single_exp_gue_concentration} to find:
    \begin{align}
    \operatorname{Pr}[\left|\tr[e^{i G t}] - \E\tr[e^{i G  t}]\right| \geq s ] &\leq \exp(-\frac{s^2}{2Nt^2}),
    \end{align}
    we can then evaluate:
\begin{align}
        \E_G |\tr D|^2  &= \int_0^\infty \Pr(|\tr D|^2\geq s)ds\\
        &= \int_0^\infty \Pr(|\tr D|\geq y)2y dy\\
        &=\left(\int_0^{|\E\tr  D|} 2ydy + \int_{|\E \tr D|}^\infty \Pr(|\tr D| \geq y)2ydy\right)\\
        &= \left(|\E\tr D|^2 + \int_{|\E \tr D|}^\infty \exp(- (y-|\E\tr D|))^2/(2Nt^2))2ydy\right)\\
        &= \left(|\E\tr D|^2 + 2Nt^2 + |\E\tr D|\sqrt{\pi 2Nt^2}\right)\\
        &\leq K_0^2t^2 + 2N t^2 +K_0t^2\sqrt{2\pi N}.
    \end{align}
    The possibilities of (4) are as such:
    \begin{enumerate}[label=\alph*.]
        \item $(\tr D)^2 (\tr (D^\dagger)^2)$
        \item $(\tr D^\dagger)^2  (\tr (D^2))$
        \item $|\tr D|^2 N$.
    \end{enumerate}
    The term from (4.c) is larger than that of (3); we find that (4.a) and (4.b) can also be bounded by (4.c)
    \begin{align}\nonumber
        |\E_G (\tr D)^2 (\tr (D^\dagger)^2)| &\leq \E_G |(\tr D)^2 (\tr (D^\dagger)^2)| \leq \E_G |\tr D|^2 |\tr (D^\dagger)^2|\leq |\tr D|^2 N,\\
        |\E_G (\tr D^\dagger)^2  (\tr (D^2))| &\leq \E_G |(\tr D^\dagger)^2 (\tr (D)^2)| \leq |\tr D|^2 N.
    \end{align}
    Finally,
    \begin{align}
        \E_G |\tr D|^2 N \leq N^2K_0^2t^2 + 2N^2 t^2 +K_0t^2N\sqrt{2\pi N}
    \end{align}
    Naively, then, it seems like (4) gives the largest contribution. However, we recall from the proof of \autoref{thm:projected_designs_GUE} that permutations that do not contract $D$ with $D^\dagger$ consecutively (and vice versa) mix the traces of $z$'s and $\ketbra{\phi}$ in the $C$ function, which means the contribution is only $\mathcal{O}(N_B)$ from the $\mathcal{Q_B}$ sum. This means that actually the maximum contribution in all the cases we looked at today can only at best be:
    \begin{align}
        \mathcal{O}(N_{AB}^{-4}\times N_B \times N_{AB}^2 K_1 t^2) = \mathcal{O}(t^2/(N_A^2N_B)). 
    \end{align}
    Overall, this leaves us with
    \begin{align}
        \E_G F^{(k)} \leq  F^{(k)}_{Haar}+ \mathcal{O}(t^2/N_A^2N_B).
    \end{align}
    This means
    \begin{align}
        \Delta^{(k)}_{rms} \leq \sqrt{\mathcal{O}\left(\frac{t^2}{N_A^2N_B}\right){N_A+k-1\choose k}} \leq \mathcal{O}\left(\frac{t N_A^{k/2-1}}{k!\sqrt{N_B}}\right).
    \end{align}
    We then recall 
    \begin{align}
        \E \norm{\rho^{(k)} - \rho^{(k)}_{Haar}}_1 \leq \E \Delta^{(k)}\leq \Delta^{(k)}_{rms}.
    \end{align}
    We apply Markov's inequality to get our design result.
\end{proof}
There are a number of interesting remarks we can make about this result. The first is that we can remove another assumption using the same proof. Previously we have required that the first trace moment perfectly evaluates to zero, but now that there is concentration coming into play on the first moment, we can also allow for some error in the time precision. Note that we weren't able to discuss this in the thermodynamic limit since everything would have scaled with $N.$
\begin{corollary}\label{cor:time_robustness}
    Let $G$ be an $N_{AB}$-dimensional random matrix drawn from the GUE that acts on $\mathcal{H}_{A}\otimes\mathcal{H}_B$. Then consider the projected ensemble formed from time evolution under this Gaussian Hamiltonian on some initial state $\ket\phi$ and measurement in the computational basis on $n_B$ qubits:
    \begin{align}
        \mathcal{E} = \{q_z, \ket\Phi_z\},
    \end{align}
    where 
    \begin{align}
        &\ket{\Phi_z} = \ket{\tilde{\Phi}_z}/\sqrt{q_z},\\
        &\ket{\tilde{\Phi}_z} = \left(\mathbb{I}_A\otimes \bra{z}_B \right)\;e^{-iGt}\ket\phi_{AB},\\
        &q_z = \braket{\tilde{\Phi}_z},\\
        &z\in\{0,1\}^{n_B}.
    \end{align}
    Then the projected ensemble $\mathcal{E}$ satisfies
    \begin{align}
        \E_G F^{(k)} \leq  F^{(k)}_{Haar}+ \mathcal{O}(t^2/N_A^2N_B).
    \end{align}
    and with a probability of $1-\mathcal{O}(1/N_A^{1-\alpha})$, it is a $\epsilon$-approximate state $k$-design for any $\epsilon < \mathcal{O}(1/N_A^\alpha)$, for any $k$ such that $N_B>N_A^{k+2\alpha+2}$, for time $t < N_A$ that satisfies $J_1(2t)/t < \mathcal{O}(1/N_{AB})$, and for $N_{AB}>4^{16}$. 
\end{corollary}
Another thing to note about the proof is that implicitly we have bounded an object physicists often like to analyze in regards to chaos: the spectral form factor, $|\tr D|^2.$ Our bound here is loose, and not sensitive to the particular features (dip, ramp, plateau) of the spectral form factor, but a tighter analysis using knowledge of the SFF of the GUE would find a better bound for some time periods. In particular, the specific features of the GUE SFF should give us better bounds in the long timescale period. Right now our results only apply up to an upper bound on $t$, although the design quality proved here holds for longer times the larger $N_B$ is compared to $N_A$. We do believe that in general, though, the design quality should hold for longer times than what we are able to prove here. Moreover at late enough times, the we expect the first trace moment to get small enough such that the exponentiated GUE projected ensemble is a good design at continuous times instead of just discrete points; this would also interact with the spectral form factor. We plan to do this more fine-grained analysis at a later date. 

Lastly, we want to make a remark that this proof again requires incredible faith in the replica trick. The use of the 4th moment is again relevant for our analytic tractability. Whether or not the replica trick can be stretched this far remains to be seen; perhaps our work could actually provide a test for specialists in numerics to analyze.


\section{From a 2k-design to a k-design by subsystem projection}\label{sec:2k-k_firstproof}
For more generic systems in which the global state may not be the result of time evolution of a GUE Hamiltonian, we also show that we may go from an approximate $2k$-design to an approximate $k$-design via projection. Our result improves upon the results and significantly simplifies the proof in \cite{Cotler2023EmergentQuantumStateDesigns}, which previously showed that projecting $n_B$ qubits from an approximate $2k$-design yields a $2k/n_B$ design. Notably, our proof in this section does not utilize the replica trick.

To make this section self-contained, let us remind readers of the notation we use: the bath subsystem $B$, consisting of $n_B$ qubits and of dimension $N_B=2^{n_B}$, is measured, obtaining a string $z\in \{0,1\}^{n_B}$. The remaining quantum state lives on system $A$, consisting of $n_A$ qubits and of dimension $N_A=2^{n_A}$. 

\begin{theorem}\label{thm:2k->k}
    Let $\Phi\in\mathcal{H}_{AB}$ be a state drawn from an $\varepsilon'$-approximate $2k$-design, i.e. for $\varepsilon' < \mathcal{O}\left(\frac{1}{N_B^{2k}N_A}\right)$:
    \begin{align}
        \norm{\E_\Phi \left(\ketbra{\Phi}\right)^{\otimes l} - \E_{\Psi\leftarrow Haar}\left(\ketbra{\Psi}\right)^{\otimes l}}_1  < \mathcal{O}\left(\frac{1}{N_B^{2k}N_A}\right), \; \forall l\leq 2k.
    \end{align}
    where $n_B > (k+1)n_A$ (equivalently $N_B > N_A^{k+1}$). Consider the projected ensemble formed from measurement in the computational basis on $n_B$ qubits:
    \begin{align}
        \mathcal{E} = \{q_z, \ket{\Phi_z}\},
    \end{align}
    where 
    \begin{align}
        &\ket{\Phi_z} = \ket{\tilde{\Phi}_z}/\sqrt{q_z},\\
        &\ket{\tilde{\Phi}_z} = \left(\mathbb{I}_A\otimes \bra{z}_B \right)\;\ket\Phi_{AB},\\
        &q_z = \braket{\tilde{\Phi}_z}, \label{eq:G5}\\
        &z\in\{0,1\}^{n_B}.
    \end{align}
    Then with probability $1-\mathcal{O}(1/N_A^{(1/2)-\alpha})$, where $\alpha \in (0,\frac{1}{2})$ is a constant, the projected ensemble $\mathcal{E}$ is an $\epsilon$-approximate quantum state $k$-design for $\epsilon<k/N_A^\alpha$ where $N_A=2^{n_A}$. In particular, with probability inverse-exponentially (in $n_A$) close to 1, the ensemble $\mathcal{E}$ is an inverse-exponentially approximate state k-design, assuming $k = \mathcal{O}(\text{poly}(n))$.
\end{theorem}

To prove this theorem, we need to first see how, on average, a draw from the global $2k$-design ensemble performs in terms of its subsystem design quality. Then, using a concentration statement, we will show that with high probability, each individual draw from the global $2k$-design will result in a $k$-design projected ensemble. 
We tackle these issues in the next two lemmas.
\begin{lemma}\label{lem:td_bound_lemma}
    For the assumptions given in \autoref{thm:2k->k}, the following holds:
\begin{equation}
         \mathbb{E}_{\Phi \sim (\epsilon',2k)\text{ design}}\Biggl\lVert \sum_{z\in\{0,1\}^{n_B}} q_z\bigl(\ket{\Phi_z}\!\bra{\Phi_z}\bigr)^{\otimes k} - \mathbb{E}_{\Psi \sim \text{Haar}}\bigl(\ket{\Psi}\!\bra{\Psi}^{\otimes k} \bigr)\Biggr\rVert_1 \leq \epsilon 
    \end{equation}
    for $\epsilon<\mathcal{O}(k/N_A^{1/2})$.
\end{lemma}
\begin{proof}
    For any $z\in \{0,1\}^{n_B}$, let $\mu_k:= \mathbb{E}_{\Psi\sim \text{Haar}}[q_z^k].$ Note that the value of this quantity does not depend on $z$. 
    
    Define 
    \begin{equation}\label{eq:Adef}
        A(\ket{\Phi}) = \sum_{z\in\{0,1\}^{n_B}} q_z\ket{\Phi_z}\!\bra{\Phi_z}^{\otimes k}  = \sum_{z\in\{0,1\}^{n_B}} \frac{\bigl(\ket{\tilde{\Phi}_z}\!\bra{\tilde{\Phi}_z}\bigr)^{\otimes k}}{q_z^{k-1}},
    \end{equation}
    so that the trace distance expression above can be rewritten as:
    \begin{equation}\label{eq:trdA}
        \mathbb{E}_{\Phi \sim (\epsilon',2k)\text{ design}}\Biggl\lVert A(\ket{\Phi}) - \mathbb{E}_{\Psi \sim \text{Haar}}A(\ket{\Psi}) \Biggr\rVert_1 .
    \end{equation}
    The equivalence of the Haar averaged terms, i.e. $\mathbb{E}_{\Psi \sim \text{Haar}}\bigl(\ket{\Psi}\!\bra{\Psi}^{\otimes k} \bigr) = \mathbb{E}_{\Psi \sim \text{Haar}}A(\ket{\Psi})$, is proven in Lemma 3 of \cite{Cotler2023EmergentQuantumStateDesigns}.
    
    Our strategy to bound this expression is inspired by the approach used by \cite{Cotler2023EmergentQuantumStateDesigns} to prove their Theorem 1, in that we will introduce an intermediate operator $B(\ket{\Phi})$, such that the trace distance in \autoref{eq:trdA} can be bounded by a triangle inequality:
    \begin{align}\label{eq:td_triangle}
        &\mathbb{E}_{\Phi \sim (\epsilon',2k)\text{ design}}\Biggl\lVert A(\ket{\Phi}) - \mathbb{E}_{\Psi \sim \text{Haar}}A(\ket{\Psi}) \Biggr\rVert_1 \notag \\ 
        &\leq \mathbb{E}_{\Phi \sim (\epsilon',2k)\text{ design}}\Biggl\lVert A(\ket{\Phi}) - B(\ket{\Phi}) \Biggr\rVert_1  + \mathbb{E}_{\Phi \sim (\epsilon',2k)\text{ design}}\Biggl\lVert B(\ket{\Phi}) - \mathbb{E}_{\Psi \sim \text{Haar}}A(\ket{\Psi}) \Biggr\rVert_1
    \end{align}

    However, our proof differs from that of \cite{Cotler2023EmergentQuantumStateDesigns} by virtue of the fact that we choose a different $B(\ket{\Phi})$. Our choice of $B(\ket{\Phi})$ both simplifies the proof significantly and improves the design quality of the projected ensemble from $k/(n_B/2)$ to $k$. Specifically, we use the following $B(\ket{\Phi})$: 
    \begin{align}
        B(\ket{\Phi}) =  \sum_{z\in\{0,1\}^{n_B}} \frac{\ket{\tilde{\Phi}_z}\!\bra{\tilde{\Phi}_z}^{\otimes k}}{\mathbb{E}_{\Phi\sim\text{Haar}}(q_z^{k-1})} = \sum_{z\in\{0,1\}^{n_B}} \frac{\ket{\tilde{\Phi}_z}\!\bra{\tilde{\Phi}_z}^{\otimes k}}{\mu_{k-1}} .
    \end{align} 
    This operator has no physical significance in the sense that it does not describe the quantum state that appears in the experiment. We choose it for its mathematical convenience in facilitating the proof. Before embarking on the detailed proof, we provide some intuition for why each of the terms in the triangle inequality is bounded. Bounding both terms crucially relies on the fact that the expectation values of polynomial functions of the unnormalized states $\ket{\tilde{\Phi}_z}$ on subsystem $A$ (obtained by projective measurements on the initial global state-design) are close to their Haar counterparts. Additionally, bounding each term involves an extra core ingredient:
    
    \textit{The first term} --- this relies on the Jensen gap $\mathbb{E}_{\Phi\sim\text{Haar}}(q_z^m) - \mathbb{E}_{\Phi\sim\text{Haar}}(q_z)^m$ being suppressed as $\mathcal{O}\left(\frac{1}{N_B^mN_A}\right)$. While the smallness of this gap reflects the strong concentration properties of the Haar measure, it is notable that we do not explicitly need to use any concentration inequalities, such as the higher-order variant of Levy's lemma used by \cite{Cotler2023EmergentQuantumStateDesigns}. Directly bounding the Jensen gap (a simple algebraic calculation) suffices.
    
    \textit{The second term} --- this relies on the closeness of the design frame potential to the Haar frame potential, via a version of the L1-L2 norm inequality expressed in terms of the frame potential. The dimension blowup is less drastic than that of the usual L1-L2 inequality, since it is just the dimension of the symmetric subspace rather than the full $k$-copy Hilbert space.\\
   
    Before diving into bounding each term in \autoref{eq:td_triangle} individually, we note the facts about the closeness of design expectations to Haar expectations that will be used for both terms. First,
    \begin{equation}\label{eq:haar_exp_qk}
        \mu_k := \E_{\Phi\sim \text{Haar}}(q_z^{k}) = \E_{\Phi\sim \text{Haar}}\left(\braket{\tilde{\Phi}_z}^k\right) =\frac{\binom{N_A+k-1}{k}}{\binom{N+k-1}{k}}
        = \prod_{j=0}^{k-1}\frac{N_A+j}{N+j}.\end{equation}
    Here $N=N_AN_B$. The analogous expectation when $\ket{\Phi}$ is drawn from an $\varepsilon'$-approximate $2k$-design is $\varepsilon'$-close to $\E_{\Phi\sim \text{Haar}}(q_z^{k})$:
    \begin{equation}\label{eq:d_to_haar_qk}
         \abs{\mathbb{E}_{\Phi \sim (\epsilon',2k)\text{-design}}(q_z^k) -  \mu_k} \leq \varepsilon'
    \end{equation}
    In fact, this closeness of the design expectation to Haar expectation holds also for the sum of $q_z^k$ over all values of $z$:
    \begin{equation}\label{eq:d_to_haar_sum_qk}
        \abs{\sum_{z\in\{0,1\}^{n_B}}\mathbb{E}_{\Phi \sim (\epsilon',2k)\text{-design}}(q_z^k) - \sum_{z\in\{0,1\}^{n_B}}\mu_k} \leq \varepsilon'
    \end{equation}
    The above can be shown by the following argument (which appears in \cite{Cotler2023EmergentQuantumStateDesigns}):
\begin{align}
        &\abs{\sum_{z\in\{0,1\}^{n_B}}\mathbb{E}_{\Phi \sim (\epsilon',2k)\text{-design}}(q_z^k) - \sum_{z\in\{0,1\}^{n_B}}\mu_k}\\
        &= \abs{\Tr\Biggl[\left(\mathbb{I}_A^{\otimes k} \otimes \sum_{z\in\{0,1\}^{n_B}} \ketbra{z}{z}_B^{\otimes k}\right) \left(\mathbb{E}_{\Phi \sim (\epsilon',2k)\text{-design}} \ketbra{\Phi}{\Phi}^{\otimes k} - \E_{\Phi\sim \text{Haar}} \ketbra{\Phi}{\Phi}^{\otimes k} \right)\Biggr]}\\
        &= \norm{\mathbb{I}_A^{\otimes k} \otimes \sum_{z\in\{0,1\}^{n_B}} \ketbra{z}{z}_B^{\otimes k}}_{\infty} \norm{\mathbb{E}_{\Phi \sim (\epsilon',2k)\text{-design}} \ketbra{\Phi}{\Phi}^{\otimes k} - \E_{\Phi\sim \text{Haar}} \ketbra{\Phi}{\Phi}^{\otimes k}}_1\\
        &\leq 1 \cdot \varepsilon'\\
        &\leq \varepsilon'.
    \end{align}
     
    Relatedly,
    \begin{equation}
        \E_{\Phi\sim \text{Haar}}\left(\left|\langle{\tilde{\Phi}_{z_1}}|\tilde{\Phi}_{z_2}\rangle\right|^{2k}\right)
        = \frac{k!\prod_{j=0}^{k-1}N_A+j}{\prod_{i=0}^{2k-1} N+i} \;\;\; \text{for }z_1\neq z_2.
    \end{equation}
    By an argument similar to that used to prove \autoref{eq:d_to_haar_sum_qk}, one can show that the analogous design quantity, summed over all $z_1$ and $z_2$, is close to the above:
\begin{equation}\label{eq:d_to_haar_2k}
        \abs{\sum_{z_1,z_2\in\{0,1\}^{n_B}} \mathbb{E}_{\Phi \sim (\epsilon',2k)\text{-design}}\left(\left|\langle{\tilde{\Phi}_{z_1}}|\tilde{\Phi}_{z_2}\rangle\right|^{2k}\right) - \sum_{z_1,z_2\in\{0,1\}^{n_B}} \E_{\Phi\sim \text{Haar}}\left(\left|\langle{\tilde{\Phi}_{z_1}}|\tilde{\Phi}_{z_2}\rangle\right|^{2k}\right)} \leq \varepsilon'
    \end{equation}
    Now we will proceed to bounding each of our trace norm expressions in the triangle inequality \autoref{eq:td_triangle}.
    
    \subsection{Bounding the first term, $\mathbb{E}_{\Phi \leftarrow 2 k \text {-design }} \| A(|\Phi\rangle)-B(|\Phi\rangle) \|_1$} 

\begin{align}
        \mathbb{E}_{\Phi \sim (\epsilon',2k)\text{-design}}\Biggl\lVert A(\ket{\Phi}) - B(\ket{\Phi}) \Biggr\rVert_1 &= \mathbb{E}_{\Phi \sim (\epsilon',2k)\text{-design}}\Biggl\lVert \sum_{z\in\{0,1\}^{n_B}} \frac{\bigl(\ket{\tilde{\Phi}_z}\!\bra{\tilde{\Phi}_z}\bigr)^{\otimes k}}{q_z^{\,k-1}} - \sum_{z\in\{0,1\}^{n_B}} \frac{\bigl(\ket{\tilde{\Phi}_z}\!\bra{\tilde{\Phi}_z}\bigr)^{\otimes k}}{\mu_{k-1}} \Biggr\rVert_1 \\ 
        &\leq \mathbb{E}_{\Phi \sim (\epsilon',2k)\text{-design}}\Biggl\lVert \sum_{z\in\{0,1\}^{n_B}} \Bigl(\frac{1}{q_z^{\,k-1}} - \frac{1}{\mu_{k-1}}\Bigr) \bigl(\ket{\tilde{\Phi}_z}\!\bra{\tilde{\Phi}_z}\bigr)^{\otimes k}\Biggr\rVert_1 \\ 
        &\leq  \mathbb{E}_{\Phi \sim (\epsilon',2k)\text{-design}} \Biggl[\sum_{z\in\{0,1\}^{n_B}} \abs{\frac{1}{q_z^{\,k-1}} - \frac{1}{\mu_{k-1}}} \;\braket{\tilde{\Phi}_z}{\tilde{\Phi}_z}^{k}\Biggr]\\
        &= \mathbb{E}_{\Phi \sim (\epsilon',2k)\text{-design}} \Biggl[ \sum_{z\in\{0,1\}^{n_B}} \abs{\frac{1}{q_z^{\,k-1}} - \frac{1}{\mu_{k-1}}}\;q_z^k \Biggr]\\
        &= \mathbb{E}_{\Phi \sim (\epsilon',2k)\text{-design}} \Biggl[ \sum_{z\in\{0,1\}^{n_B}} \abs{\,q_z - \frac{q_z^{k}}{\mu_{k-1}}\,}\Biggr]\\
        &= \mathbb{E}_{\Phi \sim (\epsilon',2k)\text{-design}} \Biggl[ \sum_{z\in\{0,1\}^{n_B}} \frac{\abs{\,q_z^{k} - q_z\,\mu_{k-1}\,}}{\mu_{k-1}}\Biggr]\\
        &= \sum_{z\in\{0,1\}^{n_B}}\frac{1}{\mu_{k-1}}\cdot\mathbb{E}_{\Phi \sim (\epsilon',2k)\text{-design}} \abs{\,q_z^{k} - q_z\,\mu_{k-1}\,}\label{eq:sum_of_devs}.
    \end{align}

    For each $z$, we can bound the above expectation as follows. First, note that by Jensen's inequality, the mean absolute deviation is upper bounded by the standard deviation:
    \begin{equation}\begin{aligned}\label{eq:jensen_mad}
        \mathbb{E}_{\Phi \sim (\epsilon',2k)\text{ design}} \abs{q_z^{k} - q_z\mu_{k-1}} \leq \sqrt{ \mathbb{E}_{\Phi \sim (\epsilon',2k)\text{ design}} \abs{q_z^{k} - q_z\mu_{k-1}}^2}
    \end{aligned}\end{equation}
    Now, by expanding the term inside the square root and using \autoref{eq:d_to_haar_qk}, we can upper bound it by its Haar counterpart plus an error term of $\mathcal{O}(\varepsilon')$:
    \begin{equation}\begin{aligned}\label{eq:designtohaar}
         \mathbb{E}_{\Phi \sim (\epsilon',2k)\text{ design}} \abs{q_z^{k} - q_z\mu_{k-1}} \leq \sqrt{\mathbb{E}_{\Phi \sim \text{Haar}} \abs{q_z^{k} - q_z\mu_{k-1}}^2 + \mathcal{O}(\varepsilon')}
    \end{aligned}\end{equation}
    The original design error $\varepsilon'$ can be chosen to be $\frac{1}{N_B^{2k}N_A}$.
    So now our task is to bound the first term in the square root,
    \begin{equation}\label{eq:haar_abs_dev}
        \mathbb{E}_{\Phi \sim \text{Haar}} \abs{q_z^{k} - q_z\mu_{k-1}}^2 = \mu_{2k}^2 - 2\mu_1\mu_{k-1} + \mu_1^2\mu_{k-1}^2
    \end{equation}    
    Instead of a messy brute-force calculation of the right-hand-side of \autoref{eq:haar_abs_dev}, we will expand the left-hand-side in a series of triangle inequalities. The latter approach has the advantage of organizing the calculation in a conceptually clearer way, highlighting the essential role played by the smallness of the Jensen gap, $\mu_m - \mu_1^m$ for $m>0$, at multiple steps of the proof. The small Jensen gap reflects the strong concentration of $q_z^m$ over the Haar measure. 

    Before bounding the expression in \autoref{eq:haar_abs_dev}, we will first bound the aforementioned gap. By Jensen's inequality, which holds for any $m$ because $q_z$ is always non-negative, $\mu_m \geq \mu_1^m$. The gap between the two sides of the inequality is
    \begin{align}
        \mu_m - \mu_1^m &= \prod_{i=0}^{m-1}\frac{N_A + i}{N_AN_B + i} - \frac{1}{N_B^m}\\
        &\leq \frac{\prod_{i=0}^{m-1}(N_A + i)}{N_A^mN_B^m} - \frac{1}{N_B^m}\\
        &\leq \frac{(N_A + m)^m}{N_A^mN_B^m} - \frac{1}{N_B^m}\\
        &= \frac{1}{N_B^m}\Biggl[\frac{(N_A + m)^m}{N_A^m} - 1\Biggr]\\
        &= \frac{1}{N_B^m}\Biggl[\left(1 + \frac{m}{N_A}\right)^m - 1\Biggr]\\
        &\leq \frac{1}{N_B^m}\Biggl[e^{\frac{m^2}{N_A}} - 1\Biggr]\\
        &\leq \mathcal{O}\left(\frac{m^2}{N_B^mN_A}\right)\label{eq:jensen_gap_bound}
    \end{align}
    where we have used the big-$\mathcal{O}$ to absorb the subleading terms for notational clarity.
    
    Now we turn to bounding the expression in \autoref{eq:haar_abs_dev}. By the triangle inequality,
\begin{align}
        \mathbb{E}_{\Phi \sim \text{Haar}} \abs{q_z^{k} - q_z\mu_{k-1}}^2
        &\leq \mathbb{E}_{\Phi \sim \text{Haar}} \Biggl[\Biggl(\abs{q_z^{k} - \mu_1^k} + \abs{\mu_1^k - q_z\mu_{k-1}}\Biggr)^2\Biggr]\\
        &= \mathbb{E}_{\Phi \sim \text{Haar}} \Biggl[\abs{q_z^{k} - \mu_1^k}^2 + \abs{\mu_1^k - q_z\mu_{k-1}}^2 + 2\abs{q_z^{k} - \mu_1^k}\abs{\mu_1^k - q_z\mu_{k-1}}\Biggr]\\\nonumber
        &\leq \underbrace{\mathbb{E}_{\Phi \sim \text{Haar}} \Biggl[\abs{q_z^{k} - \mu_1^k}^2\Biggr]}_{\text{Term (a)}} + \underbrace{\mathbb{E}_{\Phi \sim \text{Haar}} \Biggl[\abs{\mu_1^k - q_z\mu_{k-1}}^2\Biggr]}_{\text{Term (b)}} \\
        &\;\;\;\;\;\;\;+ 2\sqrt{\mathbb{E}_{\Phi \sim \text{Haar}} \Biggl[\abs{q_z^{k} - \mu_1^k}^2\Biggr] \mathbb{E}_{\Phi \sim \text{Haar}} \Biggl[\abs{\mu_1^k - q_z\mu_{k-1}}^2\Biggr]}\label{eq:termsIandII_def}
    \end{align}
    where in the final step above we have used the Cauchy-Schwartz inequality. To bound this expression it suffices to bound Terms (a) and (b), because the third term is just the square root of their product. 
    
    \paragraph{Bounding Term (a)} ---
    \begin{align}
        \mathbb{E}_{\Phi \sim \text{Haar}} \Biggl[\abs{q_z^{k} - \mu_1^k}^2\Biggr] = \mathbb{E}_{\Phi \sim \text{Haar}} \Biggl[q_z^{2k} + \mu_1^{2k} - 2q_z^k\mu_1^k\Biggr]= \mathbb{E}_{\Phi \sim \text{Haar}} \Biggl[q_z^{2k} + \mu_1^{2k} - 2\mu_k\mu_1^k\Biggr]
    \end{align}
    where we have used linearity of expectation. By Jensen's inequality, $\mu_k \geq \mu_1^k$, allowing us to upper bound the above expression by replacing $\mu_k$ in the third term by $\mu_1^k$:
    
    \begin{align}
        \mathbb{E}_{\Phi \sim \text{Haar}} \Biggl[q_z^{2k} + \mu_1^{2k} - 2\mu_k\mu_1^k\Biggr]
        \leq \mathbb{E}_{\Phi \sim \text{Haar}} \Biggl[q_z^{2k} + \mu_1^{2k} - 2\mu_1^{2k}\Biggr]
        \leq \mathbb{E}_{\Phi \sim \text{Haar}} \Biggl[q_z^{2k} - \mu_1^{2k}\Biggr]\label{eq:exp_q2k_bound}
    \end{align}
    This is simply the Jensen gap for $m=2k$. Using the bound computed in \autoref{eq:jensen_gap_bound}, we get the following bound on Term $(a)$:
    \begin{equation}\label{eq:termI_bound}
        \mathbb{E}_{\Phi \sim \text{Haar}} \Biggl[\abs{q_z^{k} - \mu_1^k}^2\Biggr] \leq \mathcal{O}\left(\frac{k^2}{N_B^{2k}N_A}\right).    
    \end{equation}
    
    \paragraph{Bounding Term (b)} --- By the triangle inequality,
    \begin{align}
        \mathbb{E}_{\Phi \sim \text{Haar}} \Biggl[\abs{\mu_1^k - q_z\mu_{k-1}}^2\Biggr]  &\leq \mathbb{E}_{\Phi \sim \text{Haar}} \Biggl[\Biggl(\abs{\mu_1^k - \mu_k} + \abs{\mu_k - q_z\mu_{k-1}}\Biggr)^2\Biggr]\label{eq:termII_TI}\\
        &= \mathbb{E}_{\Phi \sim \text{Haar}} \Biggl[\Biggl(\abs{\mu_k - \mu_1^k} + \mu_{k-1}\abs{r_k - q_z}\Biggr)^2\Biggr]\label{eq:termII_TI_r}
    \end{align}
    where in \autoref{eq:termII_TI_r} we have re-written the right-hand side of \autoref{eq:termII_TI} in terms of $r_k$, defined as the ratio $r_k := \mu_k/\mu_{k-1}$:
    \begin{equation}
        r_k := \frac{\mu_k}{\mu_{k-1}} = \frac{N_A + k-1}{N_AN_B + k-1}
    \end{equation}
    Expanding the square in \autoref{eq:termII_TI_r}, and using linearity of expectation and $\mathbb{E}_{\Phi \sim \text{Haar}}[\abs{r_k - q_z}] \leq \sqrt{\mathbb{E}_{\Phi \sim \text{Haar}}[\abs{r_k - q_z}^2]}$, gives
    \begin{align}
        \mathbb{E}_{\Phi \sim \text{Haar}} \Biggl[\abs{\mu_1^k - q_z\mu_{k-1}}^2\Biggr]\leq \underbrace{\abs{\mu_k - \mu_1^k}^2}_{\text{Term (b.i)}} + \mu_{k-1}^2\underbrace{\mathbb{E}_{\Phi \sim \text{Haar}}\Biggl[\abs{r_k - q_z}^2\Biggr]}_{\text{Term (b.ii)}}  + 2\mu_{k-1}\abs{\mu_k - \mu_1^k}\sqrt{\mathbb{E}_{\Phi \sim \text{Haar}}\Biggl[\abs{r_k - q_z}^2\Biggr]}
    \end{align}
    To bound Term (b.i), we note that it is the square of the Jensen gap for $m=k$, giving the following bound by 
    \begin{equation}\label{eq:termIIa_bound}
        \abs{\mu_k - \mu_1^k}^2 \leq \mathcal{O}\left(\frac{k^4}{N_B^{2k}N_A^2}\right)
    \end{equation}
    Next consider Term (b.ii). Let $\delta := r_k - \mu_1$, whose leading order term can be shown by explicit calculation to be $\mathcal{O}(\frac{k-1}{N_AN_B}) > 0$, so $r_k > \mu_1$. Term (b.ii) can then be bounded as follows:
    \begin{align}
        \mathbb{E}_{\Phi \sim \text{Haar}}\Biggl[\abs{r_k - q_z}^2\Biggr]&= \mathbb{E}_{\Phi \sim \text{Haar}}\Biggl[r_k^2 + q_z^2 - 2q_zr_k\Biggr]\\
        &= \mathbb{E}_{\Phi \sim \text{Haar}}\Biggl[r_k^2 + q_z^2 - 2\mu_1r_k\Biggr]\\
        &\leq  \mathbb{E}_{\Phi \sim \text{Haar}}\Biggl[r_k^2 + q_z^2 - 2\mu_1^2\Biggr]\\
        &= \mathbb{E}_{\Phi \sim \text{Haar}}\Biggl[(\mu_1 + \delta)^2 + q_z^2 - 2\mu_1^2\Biggr]\\
        &= \mathbb{E}_{\Phi \sim \text{Haar}}\Biggl[ q_z^2 - \mu_1^2\Biggr] + \delta^2 + 2\mu_1\delta.\label{eq:termIIb_bound_intermediate}
    \end{align}
    Let us bound each of these terms. Starting with $\mathbb{E}_{\Phi \sim \text{Haar}}[ q_z^2 - \mu_1^2]$, notice that this is simply the Jensen bound for $m=2$, so by \autoref{eq:jensen_gap_bound}:
    \begin{equation}\label{eq:termIIb_subbound}
        \mathbb{E}_{\Phi \sim \text{Haar}}\Biggl[ q_z^2 - \mu_1^2\Biggr] \leq \mathcal{O}\left(\frac{1}{N_B^{2}N_A}\right).
    \end{equation}
    Next, we need to bound $\delta = r_k - \mu_1$.
    \begin{align}
        \delta &= r_k - \mu_1\\
        &= \frac{N_A + k-1}{N_AN_B + k-1} - \frac{1}{N_B}\\
        &= \frac{(N_B-1)(k-1)}{N_B(N_AN_B + k-1)}\\
        &\leq  \frac{k-1}{N_AN_B + k-1}\label{eq:delta_bound}
    \end{align}
    Bounding \autoref{eq:termIIb_bound_intermediate} using \autoref{eq:termIIb_subbound} and \autoref{eq:delta_bound} gives the following bound on Term (b.ii):
    \begin{align}
        \mathbb{E}_{\Phi \sim \text{Haar}}\Biggl[\abs{r_k - q_z}^2\Biggr]\leq  \mathcal{O}\left(\frac{1}{N_B^{2}N_A}\right) +\frac{(k-1)^2}{(N_AN_B + k-1)^2} + \frac{2(k-1)}{N_B(N_AN_B + k-1)}\leq \mathcal{O}\left(\frac{k}{N_B^{2}N_A} \right)\label{eq:termIIb_bound}
    \end{align}
    to leading order. 
    
    Now we have all the ingredients, namely bounds on all the terms appearing in the various triangle inequalities used above. What remains is to put them together. By \autoref{eq:termIIa_bound} and \autoref{eq:termIIb_bound}, Term II is bounded as
    \begin{align}
        \mathbb{E}_{\Phi \sim \text{Haar}} \Biggl[\abs{\mu_1^k - q_z\mu_{k-1}}^2\Biggr]
        &\leq \abs{\mu_k - \mu_1^k}^2 + \mu_{k-1}^2\mathbb{E}_{\Phi \sim \text{Haar}}\Biggl[\abs{r_k - q_z}^2\Biggr] + 2\mu_{k-1}\abs{\mu_k - \mu_1^k}\sqrt{\mathbb{E}_{\Phi \sim \text{Haar}}\Biggl[\abs{r_k - q_z}^2\Biggr]}\\
        &\leq \mathcal{O}\left(\frac{k^4}{N_B^{2k}N_A^2} \right) + \mathcal{O}\left( \frac{k}{N_B^{2k}N_A} \right) + \mathcal{O} \left(  \frac{k^{3/2}}{N_B^{2k}N_A^{3/2}} \right)\\
        &\leq \mathcal{O}\left( \frac{k}{N_B^{2k}N_A} \right).
    \end{align}
    Plugging this bound on Term II, and the bound \autoref{eq:termI_bound} on Term I, into \autoref{eq:termsIandII_def} gives
    \begin{align}
        &\mathbb{E}_{\Phi \sim \text{Haar}} \abs{q_z^{k} - q_z\mu_{k-1}}^2\\\nonumber
        &\leq \mathbb{E}_{\Phi \sim \text{Haar}} \Biggl[\abs{q_z^{k} - \mu_1^k}^2\Biggr] + \mathbb{E}_{\Phi \sim \text{Haar}} \Biggl[\abs{\mu_1^k - q_z\mu_{k-1}}^2\Biggr]  + 2\sqrt{\mathbb{E}_{\Phi \sim \text{Haar}} \Biggl[\abs{q_z^{k} - \mu_1^k}^2\Biggr] \mathbb{E}_{\Phi \sim \text{Haar}} \Biggl[\abs{\mu_1^k - q_z\mu_{k-1}}^2\Biggr]}\\
        &\leq \mathcal{O}\left( \frac{k^2}{N_B^{2k}N_A} \right) + \mathcal{O}\left( \frac{k}{N_B^{2k}N_A} \right) + \mathcal{O}\left( \frac{k^{3/2}}{N_B^{2k}N_A} \right)\\
        &\leq \mathcal{O}\left( \frac{k^2}{N_B^{2k}N_A} \right)
    \end{align}
    By \autoref{eq:designtohaar}, this gives us the desired bound on $\mathbb{E}_{\Phi \sim (\epsilon',2k)\text{ design}} \abs{q_z^{k} - q_z\mu_{k-1}}$ for each $z$:\\
    \begin{align}
         \mathbb{E}_{\Phi \sim (\epsilon',2k)\text{ design}} \abs{q_z^{k} - q_z\mu_{k-1}}  &\leq \sqrt{\mathbb{E}_{\Phi \sim \text{Haar}} \abs{q_z^{k} - q_z\mu_{k-1}}^2 + \mathcal{O}(\varepsilon')}\\
         &\leq \sqrt{\mathcal{O}\left( \frac{k^2}{N_B^{2k}N_A} \right) + \mathcal{O}\left(\frac{1}{N_B^{2k}N_A}\right)}\\
         &\leq \mathcal{O}\left( \frac{k}{N_B^{k}N_A^{1/2}} \right)
    \end{align}
   Finally, via \autoref{eq:sum_of_devs}, this gives us the desired bound on the first term of our trace distance triangle inequality \autoref{eq:td_triangle}:
    \begin{align}
        \mathbb{E}_{\Phi \sim (\epsilon',2k)\text{ design}}\Biggl\lVert A(\ket{\Phi}) - B(\ket{\Phi}) \Biggr\rVert_1 
        &\leq \sum_{z\in\{0,1\}^{|B|}}\frac{1}{\mu_{k-1}}\cdot\mathbb{E}_{\Phi \sim (\epsilon',2k)\text{ design}} \abs{q_z^{k} - q_z\mu_{k-1}}\\
        &\leq N_B^k\; \mathcal{O}\left( \frac{k}{N_B^{k}N_A^{1/2}} \right)\\
        &\leq \mathcal{O}\left( \frac{k}{N_A^{1/2}} \right).\label{eq:td_first_bound}
    \end{align}
    
    \subsection{Bounding the second term $\mathbb{E}_{\Phi \leftarrow 2 k \text {-design }} \| B(|\Phi\rangle)-\mathbb{E}_{\Psi \sim \text { Haar }} A(|\Psi\rangle) \|_1$}
    For the second term, we will need to convert to the Schatten 2-norm. We will be employing the method of frame potentials from \autoref{sec:frame}. First, let us re-state the term we wish to bound: 
\begin{align}
        \mathbb{E}_{\Phi \sim (\epsilon',2k)\text{ design}}\Biggl\lVert B(\ket{\Phi}) - \mathbb{E}_{\Psi \sim \text{Haar}}A(\ket{\Psi}) \Biggr\rVert_1 &= \mathbb{E}_{\Phi \sim (\epsilon',2k)\text{ design}}\Biggl\lVert \frac{\sum_z\ketbra{\tilde{\Phi}_z}^{\otimes k}}{\mathbb{E}(q_z^{k-1})} - \mathbb{E}_{\Psi \sim \text{Haar}}\ketbra{\Psi}^{\otimes k}\Biggr\rVert_1
    \end{align}
    which is bounded in terms of the design and Haar frame potentials as follows:
    \begin{align}\label{eq:frame_pot_l1l2}
        \mathbb{E}_{\Phi \sim (\epsilon',2k)\text{ design}}\Biggl\lVert B(\ket{\Phi}) - \mathbb{E}_{\Psi \sim \text{Haar}}A(\ket{\Psi}) \Biggr\rVert_1  \leq  \sqrt{\frac{\E F^{(k)}}{F^{(k)}_{Haar}}-1}.
    \end{align}
    The design frame potential appearing in the numerator of \autoref{eq:frame_pot_l1l2} is
    \begin{align}
        \E_{design} F^{(k)} = \Tr ({\rho^{(k)}}^2) = \left(\frac{1}{\mathbb{E}(q_z^{k-1})}\right)^2\E_{design} \sum_{z_1,z_2}|\langle \tilde{\Phi}_{z_1}|\tilde{\Phi}_{z_2}\rangle|^{2k}
    \end{align}
    By \autoref{eq:d_to_haar_2k}, we can upper bound the design expectation of $\sum_{z_1,z_2}|\langle \tilde{\Phi}_{z_1}|\tilde{\Phi}_{z_2}\rangle|^{2k}$ by the corresponding Haar expectation plus the design error $\varepsilon'$:
\begin{align}
        \E_{design} F^{(k)} &\leq \left(\frac{1}{\mathbb{E}(q_z^{k-1})}\right)^2 \Biggl[\E_{Haar} \sum_{z_1,z_2}|\langle \tilde{\Phi}_{z_1}|\tilde{\Phi}_{z_2}\rangle|^{2k} \; + \; \varepsilon'\Biggr] \\
        &= \left(\frac{1}{\mathbb{E}(q_z^{k-1})}\right)^2\left[(N_B^2-N_B) \frac{k!\prod_{j=0}^{k-1}N_A+j}{\prod_{i=0}^{2k-1} N+i} + N_B \prod_{j=0}^{2k-1}\frac{N_A+j}{N+j} \; + \; \mathcal{O}\left(\frac{1}{N_B^{2k}N_A}\right) \right]\label{eq:frame_pot_bound}
    \end{align}
    In the second line we have broken up the inner product term into the cases $z_1\neq z_2$ and $z_1=z_2$.
    
    For the Haar frame potential in the denominator in \autoref{eq:frame_pot_l1l2}, recall that we are looking at the Haar average on the remaining subsystem $A$, so
    \begin{align}\label{eq:haar_frame_pot}
        F^{(k)}_{Haar} = {{k+N_A-1}\choose{k}}^{-1} = \frac{k!}{\prod_{j=0}^{k-1} N_A+j} .
    \end{align}
    Plugging \autoref{eq:frame_pot_bound} and \autoref{eq:haar_frame_pot} into \autoref{eq:frame_pot_l1l2}, we have
\begin{align}
        &\mathbb{E}_{\Phi \sim (\epsilon',2k)\text{ design}}\Biggl\lVert B(\ket{\Phi}) - \mathbb{E}_{\Psi \sim \text{Haar}}A(\ket{\Psi}) \Biggr\rVert_1 \\
        &\leq   \sqrt{\left(\prod_{j=0}^{k-2}\frac{N+j}{N_A+j}\right)^2\left((N_B^2-N_B) \frac{k!\prod_{j=0}^{k-1}N_A+j}{\prod_{i=0}^{2k-1} N+i} + N_B \prod_{j=0}^{2k-1}\frac{N_A+j}{N+j} \; + \; \mathcal{O}\left(\frac{1}{N_B^{2k}N_A}\right)\right) \frac{\prod_{j=0}^{k-1} N_A+j}{k!}-1}\\
        &\leq \sqrt{N_B^2(N_A+k-1)^2\frac{\prod_{j=0}^{k-2}N+j}{\prod_{i=k-1}^{2k-1}N+i}+\mathcal{O}\left(\frac{N_A^{k}}{k!N_B}\right) + \mathcal{O}\left(\frac{N_A^{2k-2}}{k!N_B^2N_A^{k-1}}\right)-1}\\
        &\leq \sqrt{\frac{N^{k+1}}{\prod_{i=k-1}^{2k-1}N+i}+\mathcal{O}\left(\frac{1}{N_A}\right)+\mathcal{O}\left(\frac{N_A^{k}}{N_B}\right)+\mathcal{O}\left(\frac{N_A^{k-1}}{k!N_B^2}\right) -1}\\
        &\leq \sqrt{1-\mathcal{O}(1/N)+\mathcal{O}\left(\frac{1}{N_A}\right)+\mathcal{O}\left(\frac{N_A^{k}}{k!N_B}\right)+\mathcal{O}\left(\frac{N_A^{k-1}}{k!N_B^2}\right)-1}\\
        &\leq \sqrt{1-\mathcal{O}(1/N)+\mathcal{O}\left(\frac{1}{N_A}\right)+\mathcal{O}\left(\frac{1}{k!N_A}\right)+\mathcal{O}\left(\frac{1}{k!N_A^{k+1}}\right)-1}\\
        &\leq \mathcal{O}\left(\frac{1}{N_A}\right)\label{eq:td_second_bound}
    \end{align}
    where we have used the assumed condition $N_B \geq N_A^{k+1}$.
    \subsection{Overall bound}
    Combining the bounds obtained on the two terms of our triangle inequality, from \autoref{eq:td_first_bound} and \autoref{eq:td_second_bound} respectively, we get the claimed result, concluding the proof of Lemma 16:
\begin{align}
         \mathbb{E}_{\Phi \sim (\epsilon',2k)\text{ design}}\Biggl\lVert \sum_{z\in\{0,1\}^{n_B}} q_z\bigl(\ket{\Phi_z}\!\bra{\Phi_z}\bigr)^{\otimes k} - \mathbb{E}_{\Psi \sim \text{Haar}}\bigl(\ket{\Psi}\!\bra{\Psi}^{\otimes k} \bigr)\Biggr\rVert_1 \leq \mathcal{O}\left(\frac{1}{N_A}\right) + \mathcal{O}\left( \frac{k}{N_A^{1/2}} \right)\leq \mathcal{O}\left( \frac{k}{N_A^{1/2}} \right).\\\qedhere
    \end{align}
\end{proof}
\begin{lemma}
    Suppose the result of Lemma 16 holds, namely:
\begin{equation}
         \mathbb{E}_{\Phi \sim (\epsilon',2k)\text{ design}}\Biggl\lVert \sum_{z\in\{0,1\}^{n_B}} q_z\bigl(\ket{\Phi_z}\!\bra{\Phi_z}\bigr)^{\otimes k} - \mathbb{E}_{\Psi \sim \text{Haar}}\bigl(\ket{\Psi}\!\bra{\Psi}^{\otimes k} \bigr)\Biggr\rVert_1 \leq \mathcal{O}\left( \frac{k}{N_A^{1/2}} \right).
    \end{equation}
    Then with probability at least $1 - \mathcal{O}\left(1/N_A^{(1/2)-\alpha}\right)$, for $\alpha \in (0,1/2)$, the trace norm deviation of any individual sample $\Phi \sim (\epsilon',2k)\text{ design}$ is bounded by $k/N_A^{\alpha}$:
\begin{equation}
        \Pr_{\Phi \sim (\epsilon',2k)\text{ design}}\Biggl[\Biggl\lVert \sum_{z\in\{0,1\}^{n_B}} q_z\bigl(\ket{\Phi_z}\!\bra{\Phi_z}\bigr)^{\otimes k} - \mathbb{E}_{\Psi \sim \text{Haar}}\bigl(\ket{\Psi}\!\bra{\Psi}^{\otimes k} \bigr)\Biggr\rVert_1 \geq \frac{k}{N_A^{\alpha}}\Biggr] \leq \mathcal{O}\left(\frac{1}{N_A^{(1/2)-\alpha}}\right)
    \end{equation}
\end{lemma}
\begin{proof}
    The desired statement follows by an immediate application of Markov's inequality to the result proved in Lemma 16.
\end{proof}
With these two lemmas, we can now present the proof of \autoref{thm:2k->k}.
\begin{proof} (Proof of \autoref{thm:2k->k})
    This follows directly from Lemmas 16 and 17.
\end{proof}


\input{PRL_2k-k}

%% file: PRL_2k-k.tex
\section{Facts about the distributions of Haar vectors}
Before we discuss the proof of \autoref{thm:optimality_2k-k_thm}, we will go over some useful facts about the distribution of Haar random state vectors; further exposition can be found in \cite{Devroye1986}. To begin with, we notice that a Haar random state vector can be created by normalizing a vector with entries that are independent Gaussian complex numbers \cite{watrous2018theory}.
\begin{lemma}\label{lem:gaussian_rep_haar}
Consider independent random variables $g_1,...,g_N\in\mathcal{CN}(0,1)$ and a set of $N$-dimensional vectors $\{\ket i\}$ that form an orthonormal basis of $\mathcal{H}$. Define the state
\begin{align}
    \ket\Phi = \sum_{i=1}^N \frac{g_i}{\sqrt{\sum_{j=1}^{N}|g_j|^2}}\ket{i}.
\end{align}
Then $\Phi$ is the Haar random state on $\mathcal{H}$.
\end{lemma}
From this, we can derive a number of distributions that correspond to various configurations of the amplitudes of the Haar state. First we notice that a sum of the squared amplitudes forms a Gamma distribution.
\begin{fact}
    Let $X = \sum_{i=1}^N |g_i|^2$. Then $X$ is distributed as Gamma(N,1), with PDF
    \begin{align}
        f(x) =  \frac{x^{N-1}e^{-x}}{\Gamma(N)}.
    \end{align}
\end{fact}
This can be derived via repeated convolution of the squared Gaussian variables' probability densities. We then find the joint distribution of a normalized set of random variable $X$; this will be helpful when we break the Haar random vector up into two subspaces \cite{ZyczkowskiSommers2001}. For now, we just mention the mathematical statement.
\begin{fact}
    Consider independent random variables $X_1,...,X_{N_B}\in \text{Gamma($N_A$,1)}$. Then define the random variable
    \begin{align}
    q_i = \frac{X_i}{\sum_{j=1}^{N_B} X_j}.
    \end{align}
    Then $q_i$ is distributed as Dirichlet($N_A$,...,$N_A$), with $N_B$ copies of $N_A$, which has the PDF
    \begin{align}
        f(q_1,...,q_{N_B-1}) = \frac{\Gamma(N_AN_B)}{\Gamma(N_A)^{N_B}}\prod_{i=1}^{N_B}q_i^{N_A-1}.
    \end{align}
\end{fact}
For notational purposes, we will let denote Dirichlet($N_A$,...,$N_A$) with $N_B$ copies of $N_A$ as Dirichlet($N_A; N_B$) for the rest of this paper. Note the PDF only depends on the first $N_B-1$ of $q_i$ since the distribution only has $N_B-1$ independent variables after normalization. In the following sections, it will also be helpful to have explicit expressions for the moments of the Dirichlet distribution, which can be found by directly integrating against the PDF. We observe:
    \begin{align}
        \mathbb{E}\left[\prod_{i=1}^{N_B} q_i^{k_i}\right] = \frac{\Gamma(N_AN_B)}{\Gamma(N_AN_B+\sum_ik_i)}\prod_{j=1}^{N_B}\frac{\Gamma(N_A+k_j)}{\Gamma(N_A)}.
    \end{align}
We also note in particular that for the amplitude of a single entry in a Haar random vector, we find a Beta distribution.
\begin{fact}
    Let $x=|\braket{0}{\Phi}|^2$ for an $N$-dimensional Haar random vector $\Phi$. Then $x$ is distributed as Beta$(1,N-1)$, with PDF
    \begin{align}
        f(x) = (N-1)(1-x)^{N-2}.
    \end{align}
\end{fact}
Lastly, using these probability distributions, we will rewrite the Haar random vector into a convenient form for our projected design calculations.
\begin{lemma}
    Consider random variable $\theta_z$ distributed uniformly on the interval $[0,2\pi)$, a set of Haar random vectors $\{\psi_z\}\in\mathcal{H}_A$, and a random variable $q_z$ from Dirichlet$(N_A; N_B)$. Define the state
    \begin{align}
        \ket\Phi = \sum_{z\in\{0,1\}^{ n_B}} \sqrt{q_z}e^{i\theta_z}\ket{\psi_z}\ket z
    \end{align}
    Then $\Phi$ is the Haar random vector on $\mathcal{H}_{AB}$.
\end{lemma}
\begin{proof}
    Let us begin with $\Phi$ the Haar random state on $\mathcal{H}_{AB}$. We will show that we can rewrite in the form given in the proof. First by \autoref{lem:gaussian_rep_haar}, we can rewrite the Haar state to
    \begin{align}
    \ket\Phi = \sum_{i=1}^{N_AN_B} \frac{g_i}{\sqrt{\sum_{j=1}^{N_AN_B}|g_j|^2}}\ket{i}.
    \end{align}
    Let us now split the state into the two different subspaces. Without loss of generality, we can choose the orthonormal states $\ket i$ and $\ket z$ to be the basis of $\mathcal{H}_A$ and $\mathcal{H}_B$, respectively.
    \begin{align}
    \ket\Phi = \sum_{z\in\{0,1\}^{n_B}}\sum_{i=1}^{N_A} \frac{g_{iz}}{\sqrt{\sum_{z}\sum_{i=1}^{N_A}|g_{iz}|^2}}\ket{i}_A\ket{z}_B.
    \end{align}
    We wish to rewrite the center sum over $i$ in terms of just a state and amplitude dependent on $z$. Now consider the unnormalized vector that includes all the center terms for each choice of $z$:
    \begin{align}
        \ket{\tilde\psi_z} = \sum_{i=1}^{N_A} \frac{g_{iz}}{\sqrt{\sum_{z}\sum_{i=1}^{N_A}|g_{iz}|^2}}\ket{i}_A.
    \end{align}
    Notice
    \begin{align}
        q_z=\left|\braket{\tilde\psi_z}{\tilde\psi_z}\right|^2 =  \frac{\sum_{i=1}^{N_A}|g_{iz}|^2}{\sum_{z}\sum_{i=1}^{N_A}|g_{iz}|^2}.
    \end{align}
    Since the $g_{iz}$ are independent complex Gaussian distributed, $q_z$ is distributed as Dirichlet($N_A$;$N_B$). This gives us the amplitude associated with $\ket{\tilde\psi_z}$. Moreover we know the phase of $\ket{\tilde\psi_z}$ must be distributed with $\theta$ as uniform around the unit circle to preserve the unitary invariance of the Haar measure. Then we can write
    \begin{align}
        \ket{\tilde\psi_z} = \sqrt{q_z} e^{i\theta} \ket{\psi_z}.
    \end{align}
    We have the distribution of the amplitudes; now what is the distribution of the normalized pure state $\ket{\psi_z}$? Recall that a post-selected Haar state remains Haar on the remaining subsystem \cite{ananth2022cryptography}. The quick way to see this is to remember that the Haar random state is unperturbed by applying any unitary, and so you could apply the unitary $R=\sum_z \ketbra{z}{z}\otimes R_z$, where $R_z$ is a Haar random unitary on $\mathcal{H}_A$, before post-selecting. As such, each $\ket{\psi_z}$ must be a Haar random state on $\mathcal{H}_A$. This accounts for every piece of the rewriting and hence concludes the proof.
\end{proof}


\section{Proof of \autoref{thm:optimality_2k-k_thm}}\label{sec:2k-k_secondproof}

We now restate \autoref{thm:optimality_2k-k_thm} and prove it.
\begin{theorem}
    For all fixed $k>0$, $N_A=2^{n_A}$, and $\varepsilon'>0$, there exists an $\varepsilon'$-approximate state $(2k-1)$-design $\Phi$ such that its projected ensemble does not form an $\varepsilon$-approximate $k$-design with probability:
    \begin{align*}\notag
    \operatorname{Pr}_{\Phi}\left[\norm{\sum_{z\in\{0,1\}^{n_B}} q_z\,
    \bigl(\ket{\Phi_z}\!\bra{\Phi_z}\bigr)^{\otimes k}-\rho^{(k)}_{Haar}}_1\geq \varepsilon\right]> 1-\mathcal{O}_{\varepsilon}\left(N_B^{-1}\right)
    \end{align*}
    for some constant $\varepsilon=\varepsilon(N_A,k)>0$.
\end{theorem}
\begin{proof}
    First we must construct our candidate design. Define $x=|\braket{0}{\Psi}|^2$ for an $d_A$-dimensional Haar random vector $\Psi$. We know from $x$ follows the Beta($1,N_A-1$) distribution, and we will denote the measure via
    \begin{align}
        d\lambda(x) = (N_A-1)(1-x)^{N_A-2}.
    \end{align}
    Then by using Gram-Schmidt orthogonalization, we can find a nonzero real polynomial $p_k(x)$ of degree $k$ such that for every nonzero real polynomial $q(x)$ of degree $t<k$, 
    \begin{align}
        \int_0^1 p_k(x)q(x)d\lambda(x) = 0.
    \end{align}
    We also stipulate that is is normalized as follows:
    \begin{align}
        \sup_{x\in[0,1]}|p_k(x)|<1
    \end{align}
    We can then define two positive probability measures on the space of pure states $\psi\in\mathcal{H}_A$. Let $\sigma$ be a random sign variable from $\{-1,1\}$, and $d\mu_{Haar}$ be the Haar distribution on $\mathcal{H}_A$. Then define 
    \begin{align}
        d\nu_\sigma(\psi) := 1 + \frac{1}{2}\sigma p_k(\psi(x))d\mu_{Haar}(\psi).
    \end{align}
    Using this moment measure, we will build our design construction. We define the following ensemble of states $\Phi$:
    \begin{align}
        \ket\Phi = \sum_{z\in\{0,1\}^{n_B}} \sqrt{q_z}e^{i\theta_z}\ket{\phi_z}\ket z,
    \end{align}
    where $\theta_z$ is a random phase distributed uniformly on the interval $[0,2\pi)$, $q_z$ is a random variable from Dirichlet$(N_A; N_B)$, $\phi_z$ are independently chosen from $\nu_\sigma$ conditioned on $\sigma$, and $\sigma$ is chosen with even odds from $\{-1,1\}$. \\\indent
    We will now show that this object is an exact global $(2k-1)$-design, and hence satisfies all $\varepsilon'>0$ approximation. Notice we already have the right amplitudes, so we just need to analyze the distribution of $\phi_z$. We define a helpful moment object $\rho_p^{(t)}$:
    \begin{align}
        \rho_p^{(t)} = \int(\ketbra{\psi}{\psi})^{\otimes t}  p_k(\psi(x)) d\mu_{Haar}(\psi).
    \end{align}
    We can then see that
    \begin{align}
        \int (\ketbra{\psi}{\psi})^{\otimes t} d\nu_{\sigma}(\psi) = \rho_{Haar}^{(t)} + \frac{1}{2}\sigma\rho^{(t)}_p.
    \end{align}
    We now show that $\rho^{(t)}_p=0$ for all $t<k$, but is nonzero for $t\geq k$. We can rewrite
    \begin{align}
        \rho_p^{(t)} = \int\mathbb{E}_{Haar}\left[(\ketbra{\psi}{\psi})^{\otimes t}\vert x\right]  p_k(x) d\lambda(x).
    \end{align}
    Moveover, recall that $\ket\psi = \sqrt{x}e^{i\alpha}\ket 0 + \sqrt{1-x}\ket y$. If we average over the phase $\alpha$ in $(\ketbra{\psi}{\psi})^{\otimes t}$, only terms with the same number of $\ket 0$ will remain, so we find
    \begin{align}
        \mathbb{E}_{Haar}\left[(\ketbra{\psi}{\psi})^{\otimes t}\vert x\right]  = \sum_{i=0}^t x^j(1-x)^{t-i} C_{t,i}.
    \end{align}
    Then we find for $t<k$,
    \begin{align}
        \rho_p^{(t)} = \sum_{i=0}^t  \int x^j(1-x)^{t-i} C_{t,i}  p_k(x) d\lambda(x) =0.
    \end{align}
    However for $t=k$, there will be a nonzero term
    \begin{align}
        \rho_p^{(t)} \supset \int x^{t} (\ketbra{0}{0})^{\otimes k} p_k(x) d\lambda(x)\propto \int  p_k(x)^2 d\lambda(x)\neq 0.
    \end{align}\\\indent
    Now we return to the whole design ensemble. Consider an arbitrary matrix element of the $t$-th moment:
    \begin{align}
        \mathbb{E} \left[\prod_{l=1}^t\Phi_{a_l,z_l}\Phi_{b_l,w_l}\right].
    \end{align}
    Here the first index is over $\mathcal{H}_A$ and the second index is over the $\ket z\in\mathcal{H}_b$. Averaging over the phases tells us this expectation value is only nonzero if every $z$ label occurs equally often among the $z_l$ as the $w_l.$ Let the multiplicities of the distinct labels be $(t_1,...,t_s),$ which is essentially a partition of $t$. Then, 
    \begin{align}
        \mathbb{E} \left[\prod_{l=1}^t\Phi_{a_l,z_l}\Phi_{b_l,w_l}\right] \propto \mathbb{E} \bigotimes_{j=1}^s \left(\ketbra{\phi_{z_j}}{\phi_{z_j}}\right)^{\otimes t_j}
    \end{align}
    We no longer worry about the coefficients because by definition they will match the same as Haar. We just need to know if partitioned tensor product of $\ket{\phi_{z_j}}$ states match with the Haar value. We find that conditioned on $\sigma$,
    \begin{align}
        \mathbb{E}\left[\left.\bigotimes_{j=1}^s\left(\ketbra{\phi_{z_j}}{\phi_{z_j}}\right)^{\otimes t_j}\right\vert\sigma\right]=\bigotimes_{j=1}^s \rho_{Haar}^{(t_j)} + \frac{1}{2}\sigma\rho^{(t_j)}_p.
    \end{align}
    When we then average over $\sigma$, we find that
    \begin{align}
        \mathbb{E}_\sigma \bigotimes_{j=1}^s \rho_{Haar}^{(t_j)} + \frac{1}{2}\sigma\rho^{(t_j)}_p=\bigotimes_{j=1}^s \rho_{Haar}^{(t_j)}.
    \end{align}
    This is because any term in the tensor product expansion with $\rho^{(t_j)}_p$ for $t_j<t$ will disappear; however, because $\sum_j t_j<2t$, there can only ever be one factor of $\rho^{(t_j)}_p$  with $t_j\geq t$, and that factor will disappear upon averaging over $\sigma.$
    As such, for $t\leq 2k-1$
    \begin{align}
        \mathbb{E} \left[\left(\ketbra{\Phi}{\Phi}\right)^{\otimes t}\right] = \rho^{(t)}_{Haar}.
    \end{align}
    Hence $\Phi$ is an exact global $(2k-1)$-design on $\mathcal{H}_{AB}$.\\\indent
    Now we will show that the projected ensemble cannot be a $k$-design. After measurement on $\mathcal{H}_B$, denote the $t$-th moment of our projected ensemble as 
    \begin{align}
        A^{(t)}(\Phi) = \sum_{z\in\{0,1\}^{n_B}} q_z\left(\ketbra{\phi_z}{\phi_z}\right)^{\otimes t}.
    \end{align}
    Conditioned on the hidden sign, we see 
    \begin{align}
        \mathbb{E}_{\nu_\sigma} \left[A^{(t)}(\Phi)\vert\sigma\right]  = \rho^{(t)}_{Haar}+\frac{1}{2}\sigma\rho^{(t)}_p:=M^{(t)}_\sigma.
    \end{align}
    We will show that $A^{(t)}(\Phi)$ strongly concentrates around this expected value, which is not equal to the Haar value for $t=k$, and hence disqualifies the projected ensemble from being an exact $k$-design with high probability. First working in the Schatten 2-norm, we find
    \begin{align}
        \mathbb{E}_{\nu_\sigma}\left[\left.\norm{A^{(t)}(\Phi)-M^{(t)}_\sigma }^2_2\right\vert q,\sigma\right] &= \sum_{z\in\{0,1\}^{n_B}} q_z^2\norm{\left(\ketbra{\phi_z}{\phi_z}\right)^{\otimes t}-M^{(t)}_\sigma}^2_2\\
        &= \sum_{z\in\{0,1\}^{n_B}} q_z^2 (1-\tr((M^{(t)}_\sigma)^2)\\
        &\leq \sum_{z\in\{0,1\}^{n_B}} q_z^2.
    \end{align}
    Using our facts about Dirichlet distribution expectations, we have 
    \begin{align}
        \mathbb{E}_q |q_z|^2 = \frac{N_A+1}{N_B(N_BN_A+1)}
    \end{align}
    Then,
    \begin{align}
        \mathbb{E}_{\nu_\sigma}\left[\left.\norm{A^{(t)}(\Phi)-M^{(t)}_\sigma }^2_2\right\vert \sigma\right] \leq \frac{N_A+1}{N_BN_A+1}.
    \end{align}
    Then using Markov's inequality, we have
    \begin{align}\nonumber
        \operatorname{Pr} \left[\left.\norm{A^{(t)}(\Phi) - M^{(t)}_\sigma}_2\geq \frac{\norm{\rho^{(k)}_p}_2}{4}\right\vert\sigma\right]\leq \frac{16(N_A+1)}{\norm{\rho^{(k)}_p}_2^2(N_BN_A+1)}.
    \end{align}
    This concentration still has a dependence on the sign $\sigma$ for the particular element of the ensemble; the statement is that the expectation concentrates around one of two peaks. We can rewrite this as a general probability statement about the distance from Haar in the trace norm. 
    \begin{align}
        \norm{A^{(t)}(\Phi) -\rho^{(t)}_{Haar}}_1&\geq \norm{A^{(t)}(\Phi) -\rho^{(t)}_{Haar}}_2\\&\geq \norm{M^{(t)}_\sigma-\rho^{(t)}_{Haar}}_2- \norm{A^{(t)}(\Phi) - M^{(t)}_\sigma}_2\\
        &\geq \frac{\norm{\rho^{(k)}_p}_2}{2} - \frac{\norm{\rho^{(k)}_p}_2}{4}\\
        &= \frac{\norm{\rho^{(k)}_p}_2}{4}.
    \end{align}
    Then we find
    \begin{align}\nonumber
        \operatorname{Pr} \left[\norm{A^{(t)}(\Phi) - \rho^{(t)}_{Haar}}_1\geq \frac{\norm{\rho^{(k)}_p}_2}{4}\right]\geq 1 -  \frac{16(N_A+1)}{\norm{\rho^{(k)}_p}_2^2(N_BN_A+1)}.
    \end{align}
    This gives us a failure probability of $1-\mathcal{O}(N_B^{-1})$ for $t=k$ for fixed $k$ and $N_A$. Taking $\varepsilon = \norm{\rho^{(k)}_p}_2/4$ concludes the proof.
\end{proof}
